\documentclass[final]{siamltex}

\usepackage{amsmath}
\usepackage{amssymb}
\usepackage{amsfonts}
\usepackage{mathptmx} 
\usepackage{times}

\usepackage{graphicx}
\usepackage{color}

\def\d{\partial}
\def\<{\leqslant}           
\def\>{\geqslant}           

\def\diag{\mathop{\rm diag}}    
\def\mR{{\mathbb R}}            
\def\Tr{{\rm Tr\,}}               
\def\rT{{\rm T}}                

\def\sP{{\mathsf P}}
\def\bL{{\mathbf L}}            
\def\bE{{\mathbf E}}            

\def\rd{{\rm d}}                

\def\x{\times}

\def\wh{\widehat}
\def\wt{\widetilde}

\def\bD{{\bf D}}                

\def\cW{{\mathcal W}}           
\def\fW{{\mathfrak W}}          

\def\cX{{\mathcal X}}           
\def\fX{{\mathfrak X}}          

\def\cov{{\bf cov}}             
\def\cN{{\cal N}}               
\def\im{{\rm im}}             
\def\cS{{\cal S}}

\def\eps{\varepsilon}           
\def\Ups{\Upsilon}              

\title{State distributions and minimum relative entropy noise sequences
    in uncertain stochastic systems:
    the discrete time case\thanks{A brief version \cite{VP_2010b} of this paper was presented at the 49th IEEE CDC in 2010. This
        work was supported by the Australian Research Council.}}

\author{Igor G. Vladimirov$^\dagger$
        \and Ian R. Petersen\thanks{School of Engineering and Information Technology,
        University of New South Wales at the Australian Defence Force Academy,
        Canberra, ACT 2600, Australia. E-mail: {\small\tt igor.g.vladimirov@gmail.com, i.r.petersen@gmail.com.}}}
\begin{document}

\maketitle

\begin{abstract}
The paper is concerned with a dissipativity theory and robust performance analysis of discrete-time stochastic systems driven by a statistically uncertain random noise. The uncertainty is quantified by the conditional relative entropy of the actual probability law of the noise with respect to a nominal product measure corresponding to a white noise sequence. We discuss a balance equation,  dissipation inequality and superadditivity property for the corresponding conditional relative entropy supply as a function of  time. The problem of minimizing the supply required to drive the system between given state distributions over a specified time horizon is considered. Such variational problems, involving entropy and probabilistic boundary conditions, are known in the literature as Schr\"{o}dinger bridge problems. In application to control systems, this minimum required conditional relative entropy supply characterizes the robustness of the system with respect to an uncertain noise.
 We obtain a dynamic programming Bellman equation for the minimum required conditional relative entropy supply and establish a Markov property of the worst-case noise with respect to the state of the system. For multivariable linear systems with a Gaussian white noise sequence as the nominal noise model and Gaussian initial and terminal state distributions, the minimum required supply is obtained using an algebraic Riccati equation which admits a closed-form solution. We propose a  computable robustness index for such systems in the framework of an entropy theoretic formulation of uncertainty and provide an example to illustrate this approach.
\end{abstract}






\pagestyle{myheadings}
\thispagestyle{plain}
\markboth{IGOR G. VLADIMIROV
        AND IAN R. PETERSEN}{STATE DISTRIBUTIONS AND MINIMUM RELATIVE ENTROPY NOISE SEQUENCES}

\section{Introduction}

Design of feedback control for stochastic systems, which is usually aimed at suppressing the effect of random disturbances on the performance of the system, often confronts the situation where the statistical characteristics of the noise are not known precisely. Such statistical uncertainty can arise both from inaccuracies in prior probabilistic information on the noise and from variability of the random environment in which the control system operates. An approach which is often practiced in optimal control design in this case (see, for example, \cite{KS_1972,LS_1995}), is to employ a relatively simple model for the noise (sometimes upon augmenting the state of the system to incorporate a noise shaping filter) and to optimize the feedback in the closed-loop system   for the case of the  nominal noise.

A different paradigm is employed by robust control approaches, such as in \cite{PUS_2000}, which are aimed at achieving ``uniformly''  guaranteed  performance  of the system over a class of uncertainties (especially, in worst-case scenarios). This is  at the expense of loosing the optimality in the nominal noise case (which often plays the role of a ``center'' of the uncertainty class). However,  the robust controller itself, and the performance of the resulting closed-loop system, depends  on the  particular description of uncertainty which was used to design them. That is, the robustness of the closed-loop system, which is secured against a particular class of uncertainties, may  be less satisfactory with respect to another class of uncertainties.

Therefore, the problem of robust performance analysis for a given closed-loop system with respect to different classes of uncertainties is important regardless of whether the system has been obtained from a robust or optimal control design methodology. More precisely, the problems of interest here are concerned with the performance deterioration of a system subject to uncertain random noise in comparison to the performance of the system when subject to the nominal noise. The statistically uncertain noise can be viewed as resulting from the actions of a hypothetical noise player who has access to the current state of the system and employs this information in generating the future noise inputs in order to drive the system away from its nominal behavior.
In this regard, an important approach,  which constitutes an important part of recent robust stochastic control and filtering theory, is  provided by  formulations of statistical uncertainty using entropy theoretic constructs
 \cite{DVKS_2001,MKV_2011,PJD_2000,PUS_2000,Petersen_2006,SVK_1994,UP_2001,VKS_1995,VKS_1996a,VKS_1996b,XUP_2008} (see also  \cite{CR_2007,DE_1997,DJP_2000} for  their
connections with the risk-sensitive control).  Although entropy and related concepts have a long history in equilibrium statistical mechanics \cite{ME_1981}, their application to robust control are more reminiscent of nonequilibrium statistical physics formulations and also have a bearing on information theory \cite{CT_2006,Gray_2009}. The deviation of the actual noise probability law from the nominal noise model,  which results in a corresponding  deviation of the system from the equilibrium probability distribution under the nominal noise, can be interpreted in terms of the  supply-storage relations of dissipativity theory \cite{Willems_1972,Willems_1972II}.

The aim of the present paper is to combine the dissipativity theory viewpoint with an entropy theoretic formulation of statistical uncertainty in order to develop a tractable \emph{robustness index} for discrete-time stochastic systems driven by an uncertain random noise. The uncertainty is quantified by the \emph{conditional relative entropy} \cite{Gray_2009} of the actual noise probability law, given the initial state  of the system, with respect to a nominal product measure corresponding to a white noise sequence, independent of the initial state. This quantity measures not only the deviation of the noise from its nominal model but also the extent to which the noise player uses knowledge of the current state of the system for future noise generation. The conditional relative entropy  can therefore be interpreted as a resource which the noise player   spends economically in performing the  role of driving the system away from its nominal behavior. This nominal behavior is characterized in terms of  the existence of a \emph{nominal invariant state distribution} which the system would have in the nominal white noise case.

Although the conditional relative entropy supply is apparently less symmetric in time than the unconditional relative entropy,  it satisfies a balance equation which involves time reversal through a Bayesian term \cite{Bernardo,Hatanaka}. A related dissipation inequality describes the influence of the conditional relative entropy supply for  the noise player on the deviation of the system from the nominal invariant state distribution. The  deviation of the system from the nominal invariant state distribution  is also measured in relative entropy terms and plays the role of a \emph{storage function}. As a function of  time, the conditional relative entropy supply is \emph{superadditive} \cite{Schilling_2005} in contrast   to its deterministic counterpart in \cite{Willems_1972,Willems_1972II} (which is additive as the integral of a supply rate over the time interval). However, additivity is recovered for a class of noise sequences
which are Markov with respect to the state of the system and play an important role as economical noise strategies.

A problem of minimizing the conditional relative entropy supply required for the noise player to
drive the system between given initial and terminal state distributions over a specified time horizon  is then considered. Variational problems,  which are concerned with entropy minimization under such probabilistic boundary conditions, are known as  Schr\"{o}dinger bridge problems \cite{Beghi_1996,PT_2010}. These problems are usually treated in the  context of reciprocal processes, that is, Markov random fields on the time axis  \cite{Jamison_1974}; see also \cite{Beghi_1994,Blaq_1992,DaiPra_1991,Mikami_1990,VP_2010a} for continuous time formulations. Such problems have also been studied for quantum
systems \cite{Beghi_2002} using the formalism of stochastic mechanics
\cite{Nelson_2001}, and conventional quantum mechanical settings \cite{PT_2010}.
In application to robust performance analysis, the minimum required conditional relative entropy supply characterizes the robustness of the system with respect to the  uncertain noise. Indeed, the larger is the  required supply, the more ``sluggish'' the system is with respect to the actions of the noise player.
 We obtain a dynamic programming Bellman equation for the minimum required conditional relative entropy supply and establish the Markov property of the corresponding worst-case noise with respect to the state of the system.

A related {\it state
distribution tracking} problem leads to the minimum conditional relative entropy supply rate (per time step), which is  required for the noise player to maintain the system in a given state distribution. In combination with a \emph{loss functional} (which measures the system performance deterioration associated with the deviation from the nominal invariant state distribution), the minimum supply rate, required to achieve a specified level of the system performance loss, provides  a useful robustness index.

The specialization of the above results to the case of  multivariable linear systems with a white  Gaussian nominal noise sequence and Gaussian initial and terminal state distributions, allows the minimum required supply to be  determined using an algebraic Riccati equation which admits a closed-form solution. For a class of one-step reachable linear systems, in the framework of the entropy theoretic description of uncertainty,  we propose a particular robustness index associated with the increase in a weighted second moment of the state variables as the loss functional. Similar, though different  ideas, which combine the second moment increase with entropy theoretic formulations of uncertainty, can be found in \cite{CR_2007,DVKS_2001,MKV_2011,Petersen_2006,SVK_1994,VDK_2006,VKS_1995,VKS_1996a}.
The computation of the robustness index is reduced to solving two coupled algebraic equations in a matrix and a scalar parameter, which can be carried out numerically  by using, for example,  homotopy methods. As an illustration, we provide an explicit calculation of the robustness index for one-dimensional linear systems.

The paper is organized as follows. Section~\ref{sec:system}
specifies the class of uncertain stochastic systems being considered.
Section~\ref{sec:equilibrium}   describes the nominal white noise model and the associated nominal invariant state distribution of the system. Section~\ref{sec:cond_rel_ent} specifies a measure of statistical uncertainty as the conditional relative entropy of the actual noise with respect to the nominal noise.
Section~\ref{sec:dissipation} discusses a dissipation inequality
and a superadditivity property for the conditional relative entropy supply and
introduces Markov noise strategies. Section~\ref{sec:markovization} describes a
procedure which leads to a Markov noise strategy with a decreased conditional relative
entropy while preserving the state distributions of the system.
Section~\ref{sec:Bellman} employs this procedure to establish a dynamic programming Bellman
equation for the minimum conditional relative entropy supply, required to drive the system between given initial and terminal state distributions,  and introduces a system robustness index.
Sections
\ref{sec:state_dist_contr} to \ref{sec:example} are concerned with
the case of linear dynamics  and a white Gaussian nominal noise sequence. Section \ref{sec:state_dist_contr}
establishes conditions for the reachability of Gaussian state distributions  
of the linear system. Section~\ref{sec:min_req_supp}
reduces the problem of computing the minimum required supply for the case of Gaussian boundary conditions to an algebraic Riccati
equation. A closed-form solution of this equation
is given in
Section~\ref{sec:Ric} and is used  in Section \ref{sec:robust_meas} for computing the robustness index for a class of linear systems. Section~\ref{sec:example} provides an example which explicitly calculates the robustness index for a one-dimensional linear system.

\section{Stochastic systems with statistically uncertain noise}
\label{sec:system}

We consider a dis\-crete-time system with a state signal $X:=(X_k)_{k\> 0}$, driven by a noise input $W:=(W_k)_{k\> 0}$. In order to capture various special cases in a general formulation, the values $X_k$ and $W_k$ of these signals (at the $k$th time step)  are assumed to belong to Polish (complete separable metric)  spaces $\cX$ and $\cW$, endowed  with Borel $\sigma$-algebras $\fX$ and $\fW$, respectively.  The dynamics of the system in the \emph{state space} $\cX$ are  governed by a time-invariant equation
\begin{equation}
\label{system}
    X_{k+1} = f(X_k, W_k),
    \qquad
    k =0,1,2,\ldots,
\end{equation}
where $f: \cX \x \cW \to \cX$ is a given Borel measurable \emph{one-step state transition map}.
Thus, the states of the system at any two moments of time are related by
\begin{equation}
\label{FG}
    X_t
    =
    F_{t-s}(Y_{s:t-1})
    =
    F_{t-s}(X_s, W_{s:t-1})
    =
    F_{t-s}(X_s, W_s, \ldots, W_{t-1}),
    \qquad
    0\< s\< t.
\end{equation}
Here, $F_k: \cX \x \cW^k \to \cX$ denotes the  \emph{$k$ step state
transition map}, which satisfies the recurrence relation
\begin{equation}
\label{Fnext}
    F_{k+1}(x_0, w_0,\ldots, w_k) = f(F_k(x_0, w_0,\ldots, w_{k-1}), w_k)
\end{equation}
for all $x_0\in \cX$ and $w_0, \ldots, w_k \in \cW$,
with the initial condition that $F_0$ is the identity map on the state space
 $\cX$. Also, for the time interval $[s,t]$,
 \begin{equation}
\label{Yst}
    Y_{s:t}
    :=
    (X_s, W_{s:t})
    =
    (X_s, W_s, \ldots, W_t)
\end{equation}
denotes the \emph{state-noise sequence} which is formed from the initial state $X_s$ of the system and the \emph{noise sequence}
\begin{equation}
\label{Wst}
     W_{s:t}
     :=
     (W_s, \ldots, W_t).
\end{equation}
%
%
%
%
%
%
%
Randomness is introduced into the system (\ref{system}) by assuming that the initial state $X_0$ and the noise sequence $W$ are random elements. Their joint probability distribution $\sP_{X_0, W}$  is a probability measure\footnote{We denote by $\sP_{\xi}$ the probability distribution of a random element $\xi$, and by $\sP_{\xi \mid \eta}$ the conditional probability distribution of $\xi$ with respect to another random element $\eta$, with $\xi$ and $\eta$ taking values in Polish spaces. Thus, $\sP_{\xi\mid \eta}(S\mid y)$ is  a probability measure  of a Borel set $S$ and a Borel measurable function of $y$. The joint probability distribution of $\xi$ and $\eta$ is denoted by $\sP_{\xi,\eta}$.} on the measurable  space $(\cX\x \cW^{\infty}, \fX\x \fW^{\infty})$.
 Accordingly, the state sequence $X$ is a $\cX^{\infty}$-valued random element with a probability distribution $\sP_X$
 on $(\cX^{\infty}, \fX^{\infty})$. Since $X$ depends on $X_0$ and $W$ in a deterministic way, as described by (\ref{FG}) and (\ref{Fnext}), with the map $(X_0, W)\mapsto X$ being completely specified by the one-step state transition map $f$,   the probability distribution  $\sP_X$ can be expressed in terms of $\sP_{X_0, W}$. In particular,  consider the \emph{state distribution}
 \begin{equation}
\label{Pkdef}
    P_k
    :=
    \sP_{X_k}
\end{equation}
of the system at time $k$, that is, an appropriate marginal probability distribution of $X_k$  on the measurable space $(\cX, \fX)$ which corresponds to $\sP_X$. In view of (\ref{FG}), the state distributions (\ref{Pkdef}) are related to the probability distributions
\begin{equation}
\label{Qst}
    Q_{s,t}
    :=
    \sP_{Y_{s:t-1}}
    =
    \sP_{X_s, W_{s:t-1}}
    =
    \sP_{X_s, W_s, \ldots, W_{t-1}}
\end{equation}
of the state-noise sequences (\ref{Yst}), that is, the joint probability distributions of  $X_s$ and $W_{s:t-1}$ on $(\cX\x \cW^{t-s}, \fX\x \fW^{t-s})$:
\begin{equation}
\label{PQ}
    P_t
    =
    Q_{s,t}
    \circ
    F_{t-s}^{-1}.
\end{equation}
The right-hand side of (\ref{PQ}) is the image measure \cite[pp. 51--52]{Schilling_2005} of the probability distribution $Q_{s,t}$ under the $t-s$ step state transition map $F_{t-s}$, with $F_k^{-1}(S):= \{y\in \cX\x \cW^k:\ F_k(y)\in S\}$ the pre-image of a set $S \in \fX$.
%
In turn, $Q_{s,t}$ in (\ref{Qst}) is completely specified by the initial state distribution $P_s$ and the conditional probability distribution $\sP_{W_{s:t-1}\mid X_s}$ of the noise sequence $W_{s:t-1}$ given $X_s$ in view of the chain rule for probability measures
\begin{equation}
\label{QPP}
    Q_{s,t}(\rd x \x \rd w) =  P_s(\rd x) \sP_{W_{s:t-1}\mid X_s}(\rd w \mid x),
    \qquad
    x \in \cX,
    \
    w \in \cW^{t-s}.
\end{equation}
The equations (\ref{system}) may describe the dynamics of a closed-loop  system obtained by applying a given feedback controller to a  given plant, in which case $X$ incorporates both the plant and controller state variables. Then, the plant is subject to an external random noise $W$.  The design of such a controller often employs a relatively simple statistical model for the noise and is aimed at suppressing the influence of the  noise on the closed-loop system performance. Although the \emph{nominal} noise model is not guaranteed to be accurate, the feedback is usually developed  so as to make the system ``well-behaved'' at least under the nominal noise (for example, by an appropriate choice of the map $f$ in (\ref{system})). Whereas the meaning of this depends on a specific control context, the property of being well-behaved (which is pursued by the control designer)  is understood here as the existence of an invariant probability measure for the system state sequence $X$ under the nominal noise.

\section{Nominal noise model and nominal invariant state distribution}
\label{sec:equilibrium}

A typical nominal noise model is that $W$  is a ``white noise'' sequence of independent identically distributed random elements which are also independent of $X_0$.

\begin{definition}
\label{def:nom_scenario}
Suppose $R$ is a given probability measure on the measurable space $(\cW, \fW)$. The noise $W$ is called nominal if $W_0, W_1, \ldots$ are independent $R$-distributed random elements, independent of the initial state of the system $X_0$, so that the corresponding conditional probability distribution  is a product measure
\begin{equation}
\label{Rinf}
    \sP_{W\mid X_0}^*
    =
    R^{\infty} = R\x R\x \ldots.
\end{equation}
\end{definition}

The probability measure $\sP_{W\mid X_0}^*$ in (\ref{Rinf}), under which the noise $W$ has a simple  statistical structure (specified completely by the \emph{nominal}  marginal distribution  $R$ of $W_k$),   plays the role of a model for the unknown \emph{actual noise} probability measure $\sP_{W\mid X_0}$.
Under the nominal noise defined in  Definition~\ref{def:nom_scenario}, the state sequence $X$ is a homogeneous Markov chain with transition probability measure
\begin{equation}
\label{G}
    G(S\mid x)
    :=
    R
    (f(x,\cdot)^{-1}(S)),
    \qquad
    S \in \fX,\
    x \in \cX,
\end{equation}
where $f(x,\cdot)^{-1}(S):= \{w \in \cW:\ f(x,w)\in S\}$.
In this case, the state distributions $P_k$ from (\ref{Pkdef}) satisfy the recurrence equation
\begin{equation}
\label{PkTS}
    P_{k+1}(S)
    =
    \int_{\cX}
    G(S\mid x)
    P_k(\rd x)
    =
    (P_k \x R)(f^{-1}(S)),
\end{equation}
where $f^{-1}(S):= \{(x,w)\in \cX\x \cW:\ f(x,w)\in S\}$ is the pre-image of a set $S \in \fX$ under the one-step state transition map $f$.  An invariant measure for the Markov chain $X$ is a probability measure $P_*$ on $(\cX, \fX)$ which is a fixed point of the linear integral operator described by the right-hand side of (\ref{PkTS}). That is,
\begin{equation}
\label{P*}
    P_*
    =
    (P_* \x R)
    \circ
    f^{-1}.
\end{equation}
By induction, (\ref{PQ}) allows (\ref{P*}) to be extended to the image measure under the $k$ step state transition map $F_k$ in (\ref{FG}) and (\ref{Fnext}) as
\begin{equation}
\label{P*push}
    P_* = (P_* \x R^k) \circ F_k^{-1},
    \qquad
    k>0.
\end{equation}

\begin{definition}
\label{def:nom_inv_state_distr}
An invariant measure $P_*$ of the Markov chain $X$ under the nominal noise from Definition \ref{def:nom_scenario} in the sense of (\ref{P*}) is referred to as a nominal invariant state distribution for the system.
\end{definition}

In what follows, we assume that a nominal invariant state distribution $P_*$ for the system exists, though is not necessarily unique.   Any such $P_*$ is an equilibrium point for the state distributions $P_0, P_1, \ldots$ of the system, governed by (\ref{PkTS}) under the nominal noise.
 General criteria for the existence of invariant measures for Markov chains are beyond the scope of the  present paper. However,  we will describe a version of Harris's theorem from \cite{Hairer_2010,MT_1993}, which guarantees the existence and uniqueness of an invariant  probability measure. In application to our specific context, the sufficient conditions are as follows. Suppose there exist a Borel measurable function $V: \cX \to \mR_+$ and constants $0\<  q  <1$ and $r  \> 0$ such that the inequality
\begin{equation}
\label{A1}
    \bE V(f(x,\omega))
    :=
    \int_{\cW}
    V(f(x,w))
    R(\rd w)
    =
        \int_{\cX}
    V(y)
    G(\rd y\mid x)
    \< q  V(x) + r
\end{equation}
holds for all $x \in \cX$. Here, the expectation $\bE(\cdot)$ is taken over an $R$-distributed random element $\omega$  with values in $\cW$, and $G$ is the Markov transition kernel (\ref{G}) of the state sequence $X$ under the nominal noise.  Also, suppose
\begin{equation}
\label{A2}
    \sup_{x,y\in \cX_v}\
    \sup_{g: \cX \to [-1,1]}|\bE (g(f(x,\omega)) - g(f(y,\omega)))|<2
\end{equation}
for any $v> 0$, where $\cX_v:=\{x \in \cX:\ V(x)\< v\}$ denotes the corresponding sub-level set of the function $V$ from (\ref{A1}), and the second supremum is taken over Borel measurable real-valued functions $g$ on $\cX$ whose absolute value does not exceed one. The left-hand side of (\ref{A2}) is the diameter of the set $\{G(\cdot \mid x):\ x\in \cX_v\}$ in the sense of the total variation distance between probability measures \cite{Shiryaev}. Then, in view of \cite[Theorem 3.6 on p. 13]{Hairer_2010},  the conditions (\ref{A1}) and (\ref{A2}) (which correspond to \cite[Assumptions 3.1 and 3.4 on p. 12]{Hairer_2010}) imply that the system (\ref{system}) has a unique nominal invariant state distribution $P_*$.

The \emph{actual} conditional distribution $\sP_{W\mid X_0}$ of the noise may differ from its nominal model (\ref{Rinf}). In particular, there can be statistical \emph{dependence} between $W_k$'s at different  times or between the noise $W$ and the initial state of the system $X_0$. Also,  the marginal distribution of $W_k$ may differ from $R$ even if $W$ is indeed a white noise sequence.   The  discrepancy between the true $\sP_{W\mid X_0}$ and its nominal model, present in all these cases,  is interpreted as \emph{statistical uncertainty} in  the noise $W$.

The dependence of the conditional distribution of the future noise on the current state of the system (which depends on the past history of the noise) can arise in the case of a ``colored'' noise whose values at different moments of time are statistically dependent. Without specifying a mechanism for the memory effects in the random environment which produce such a noise\footnote{A discussion of the generation of such noise can be found, for example,  in physics literature on open systems \cite{BP_2006}.},   we will interpret the conditional probability distributions $\sP_{W_{s:t}\mid X_s}$ in (\ref{QPP}) as the strategy of a hypothetical \emph{noise player} who opposes the control designer.  More precisely, it is assumed that the noise player  has access to the current state $X_s$ of the system at any moment of time $s$ and  uses this information in generating the future noise inputs $W_s, W_{s+1}, \ldots$ so as to make the system deviate from the nominal behavior described in Section~\ref{sec:equilibrium}. In particular, this process
can be viewed as the noise player
aiming to drive the actual state distribution $P_t$ of the system away from the nominal invariant state
distribution $P_*$. That is, the noise player aims to drive the state distribution away from the probabilistic equilibrium of the system under the nominal noise.
The extent, to which the actual probability distribution  $\sP_X$ of the state sequence differs from the probability law of a Markov chain with the transition kernel (\ref{G}) and invariant measure $P_*$, depends on the amount of statistical uncertainty in the noise.

\section{Conditional relative entropy to quantify statistical uncertainty}
\label{sec:cond_rel_ent}

Similarly to stochastic robust control settings such as in
\cite{PJD_2000,UP_2001,XUP_2008}, the deviation of the conditional noise
distribution $\sP_{W\mid X_0}$ from its nominal model (\ref{Rinf})  will be quantified in terms of the \emph{conditional relative
entropy} \cite[Section~5.3]{Gray_2009}.

Recall that for two conditional probability distributions  $\sP_{\xi \mid \eta}$ and $\sP_{\xi \mid \eta}^*$ of random elements $\xi$ and $\eta$ with values in Polish spaces $\cS_1$ and $\cS_2$,  the conditional relative entropy of $\sP_{\xi \mid \eta}$ with respect to $\sP_{\xi \mid \eta}^*$  is defined as
\begin{align}
\nonumber
    \bD(\sP_{\xi \mid \eta} \| \sP_{\xi \mid \eta}^*)
    := &
    \bE \ln \varphi(\xi \mid \eta)
    =
    \int_{\cS_1\x \cS_2}
    \ln \varphi(x\mid y)
    \sP_{\xi, \eta}(\rd x \x \rd y)\\
\nonumber
    = &
    \int_{\cS_2}
    \Big(
    \int_{\cS_1}
    \bL(\varphi(x\mid y))
    \sP_{\xi\mid \eta}^*(\rd x \mid y)
    \Big)
    \sP_{\eta}(\rd y)\\
\label{bD}
    = &
    \int_{\cS_2}
    \bD_0
    (
        \sP_{\xi \mid \eta}(\cdot \mid y)
        \|
        \sP_{\xi \mid \eta}^*(\cdot\mid y)
    )
    \sP_{\eta}(\rd y),
    \qquad
    \varphi(x\mid y):=         \frac{\sP_{\xi \mid \eta}(\rd x \mid y)}
        {\sP_{\xi \mid \eta}^*(\rd x \mid y)},
\end{align}
where the expectation is taken over the joint probability distribution $\sP_{\xi,\eta}$ of $\xi$ and $\eta$, associated with $\sP_{\xi\mid \eta}$ by the chain rule $\sP_{\xi, \eta}(\rd x \x \rd y) = \sP_{\xi \mid \eta}(\rd x \mid y) \sP_{\eta}(\rd y)$, and the functional $\bD_0$ is described below.
Here, the function
\begin{equation}
\label{bL}
    \bL(p)
    :=
    p\ln p
\end{equation}
is defined on $\mR_+$,
with the standard convention that $\bL(0) = 0$.  Also, $\varphi: \cS_1\x \cS_2\to \mR_+$  in (\ref{bD}) is a Borel measurable  function, which, for any fixed but otherwise arbitrary value of its second argument $y \in \cS_2$, describes the Radon-Nikodym derivative \cite{P_1967,Schilling_2005,SG_1977} of the probability measure $\sP_{\xi \mid \eta}(\cdot \mid y)$ with respect to the reference probability measure $\sP_{\xi \mid \eta}^*(\cdot \mid y)$, so that
 $\sP_{\xi\mid \eta}(S\mid y) = \int_S \varphi(x\mid y)\sP_{\xi\mid \eta}^*(\rd x \mid y)$ for any Borel subset $S \subset \cS_1$.
 This conditional probability density function (PDF) $\varphi$ exists if and only if the first measure is absolutely continuous
 with respect to the second one:
\begin{equation}
\label{PllP}
    \sP_{\xi \mid \eta}(\cdot \mid y)
    \ll
    \sP_{\xi \mid \eta}^*(\cdot \mid y).
\end{equation}
That is, for all Borel subsets $S\subset \cS_1$, the fulfillment of $\sP_{\xi \mid \eta}^*(S \mid y)=0$ implies $\sP_{\xi \mid \eta}(S \mid y)=0$.
The functional $\bD_0$ in (\ref{bD}), which is distinguished from $\bD$, describes the \emph{unconditional} relative entropy
\begin{equation}
\label{bD0}
    \bD_0
    (
        \sP
        \|
        \sP_*
    )
=
    \int_{\cS}
    \ln \varphi(x)
    \sP(\rd x)    =
    \int_{\cS}
    \bL(\varphi(x))
    \sP_*(\rd x),
    \qquad
    \varphi(x):= \sP(\rd x)/\sP_*(\rd x),
\end{equation}
for probability measures $\sP\ll \sP_*$ on a common Polish space $\cS$
with an appropriate  PDF $\varphi: \cS\to \mR_+$. The conditional relative entropy $\bD(\sP_{\xi \mid \eta} \| \sP_{\xi \mid \eta}^*)$ in (\ref{bD}) is well-defined if the conditional absolute continuity (\ref{PllP}) holds for $\sP_{\eta}$-almost all values $y \in \cS_2$ of the random element $\eta$. It follows from the properties of  relative entropy \cite{CT_2006,Gray_2009} that both functionals $\bD$ and $\bD_0$ are always nonnegative and vanish only on equal measures (so that, in particular,  $\bD_0(\sP\|\sP_*) = 0$ if and only if $\sP = \sP_*$).

Now, when quantifying the deviation of the actual conditional probability distribution  $\sP_{W_{s:t-1}\mid X_s}$ of the noise sequence $W_{s:t-1}$ on the time interval $[s,t)$ from its nominal model $R^{t-s}$,  the conditional relative entropy (\ref{bD}) takes the form
\begin{eqnarray}
\nonumber
    E_{s,t}
    & := &
    \bD(
        \sP_{ W_{s:t-1} \mid X_s}
        \|
        R^{t-s}
    )
    =
    \bE
    \ln
    \varphi_{s,t}(W_{s:t-1} \mid X_s)\\
\nonumber
    & = &
    \int_{\cX \x \cW^{t-s}}
    \ln \varphi_{s,t}(w \mid x)
    Q_{s,t}(\rd x \x \rd w)\\
\nonumber
    & = &
    \int_{\cX}
    \Big(
        \int_{\cW^{t-s}}
        \bL(\varphi_{s,t}(w \mid x))
        R^{t-s}(\rd w)
    \Big)
    P_s(\rd x)\\
\label{Est}
    & = &
    \int_{\cX}
    \bD_0(\sP_{W_{s:t-1}\mid X_s}(\cdot \mid x) \| R^{t-s})
    P_s(\rd x),
    \qquad
        \varphi_{s,t}(w\mid x)
    :=
    \frac{\sP_{W_{s:t-1} \mid X_s}(\rd w \mid x)}{R^{t-s}(\rd w)},
\end{eqnarray}
where the expectation is taken over the probability
distribution $Q_{s,t}$ of the state-noise sequence $Y_{s:t-1}$ from (\ref{Qst}) and (\ref{QPP}). Here, the distribution of the noise sequence $W_{s:t-1}$, conditioned on $X_s$,  is assumed to be
absolutely continuous with respect to the corresponding nominal distribution $R^{t-s}$ in the sense that
\begin{equation}
\label{cond_abs_cont}
    \sP_{W_{s:t-1}
        \mid
        X_s}(\cdot \mid x)
    \ll
    R^{t-s}
    \quad
    {\rm  for}\
    P_s{\rm-almost\ all}\ x \in \cX.
\end{equation}
This ensures that the conditional PDF $\varphi_{s,t}: \cW^{t-s}\x \cX\to \mR_+$  in (\ref{Est})
exists and the quantity $E_{s,t}$ is well-defined.
Further discussion will be concerned with a class of ``admissible'' probability distributions for the noise as specified below.

\begin{definition}
\label{def:good}
A noise sequence $W$, which drives the system dynamics (\ref{system}), is called admissible if the conditional  probability distribution  $\sP_{W_{s:t-1}\mid  X_s}$ satisfies  (\ref{cond_abs_cont}) for any  times $0\< s< t$.
\end{definition}

The conditional relative entropy $E_{s,t}$ in (\ref{Est}), which is always nonnegative,     vanishes for all
$0\< s < t$ if and only if the noise sequence $W$ is
$R^{\infty}$-distributed and independent of the  initial  state $X_0$.
In what follows, when considering the system on a time interval $[s,t)$, we will always assume that the distribution of the initial state $X_s$ is  absolutely continuous with respect to the nominal invariant state distribution $P_*$. That is,
\begin{equation}
\label{PP}
    P_s \ll P_*.
\end{equation}
%
%
In view of the chain rule (\ref{QPP}), the fulfillment of conditions (\ref{cond_abs_cont}) and (\ref{PP})  implies that the actual probability distribution $Q_{s,t}$ of the state-noise sequence from (\ref{Qst}) is absolutely continuous with respect to the corresponding product measure:
\begin{equation}
\label{allow}
    Q_{s,t}
    \ll
    P_* \x R^{t-s}.
\end{equation}
Note that (\ref{allow}) implies that the property (\ref{PP}) will  be inherited by the subsequent state distribution $P_t$. Indeed, since     $P_t$ and
$
    P_* = (P_* \x R^{t-s})\circ F_{t-s}^{-1}
$
are the image measures of $Q_{s,t}$ and $P_* \x R^{t-s}$ under the same map $F_{t-s}$ in view of (\ref{PQ}) and (\ref{P*push}), then  (\ref{allow}) implies that $P_t \ll P_*$. Therefore, if $P_0\ll P_*$, then for any admissible noise $W$ in the sense of Definition~\ref{def:good}, the property $P_t\ll P_*$ holds for any $t>0$.

Although  (\ref{Est})  requires only
the conditional absolute continuity condition (\ref{cond_abs_cont}) for the noise,  the additional absolute continuity (\ref{PP}) for the state distributions  will play a role in Section~\ref{sec:dissipation}. Under the conditions (\ref{cond_abs_cont}) and (\ref{PP}), the chain rule (\ref{QPP}) allows the PDF of $Q_{s,t}$ with respect to the reference measure $P_* \x R^{t-s}$ in (\ref{allow}) to be  factorized as
\begin{equation}
\label{QPR}
    \frac{Q_{s,t}(\rd x \x \rd w)}{P_*(\rd x)R^{t-s}(\rd w)}
    =
    \varpi_s(x)\varphi_{s,t}(w\mid x),
    \qquad
    x \in \cX,\
    w \in \cW^{t-s}.
\end{equation}
Here, $\varphi_{s,t}$ is the conditional PDF of the noise sequence $W_{s:t-1}$ given $X_s$ from (\ref{Est}), and $\varpi_s: \cX\to \mR_+$ is the PDF of the actual state distribution $P_s$ with respect to the nominal invariant state distribution $P_*$:
\begin{equation}
\label{varpi}
    \varpi_s(x)
    :=
    P_s(\rd x)/P_*(\rd x),
    \qquad
    x\in \cX.
\end{equation}
In what follows, we will study several variational problems which involve the conditional relative entropy (\ref{Est}).
The quantity $E_{s,t}$, which is   a measure of deviation from the nominal noise model (\ref{Rinf}), can be regarded as a resource which the noise player would prefer to spend economically in performing the  role of driving the system away from the nominal invariant state distribution.

\section{Conditional relative entropy balance equation and  dissipation inequality}
\label{sec:dissipation}

For the purposes of the subsequent sections, we will now discuss several properties of the conditional relative entropy $E_{0,t}$, defined in (\ref{Est}), starting with its decomposition which employs time reversal and Bayesian analysis ideas \cite{Bernardo}.
Let $\sP_{Y_{0:t-1} \mid X_t}^*$ denote the conditional (given $X_t$) probability distribution  which the
state-noise sequence $Y_{0:t-1}$  would have if the system (\ref{system}) were initialized at the nominal invariant state distribution $P_*$ and were subjected to the nominal noise $W$ in the sense of Definition~\ref{def:nom_scenario} (in which case, the \emph{unconditional} probability distribution $Q_{0,t}$ of $Y_{0:t-1}$ would be $P_* \x R^t$).
The fact that $Y_{0:t-1}$, associated with the time interval $[0,t)$, is conditioned here on the \emph{terminal} state of the system $X_t = F_t(Y_{0:t-1})$ under the nominal noise, with $F_t$ the $t$ step state transition map from (\ref{FG}), motivates the following definition.

\begin{definition}
\label{def:nom_post}
The conditional probability  distribution $\sP_{Y_{0:t-1} \mid X_t}^*(\cdot \mid x)$ of a $P_*\x R^t$-distributed random element $\eta$, conditioned on $ F_t(\eta) = x$, is called the nominal posterior distribution of the state-noise sequence $Y_{0:t-1}$.
\end{definition}

Note that the nominal posterior distribution $\sP_{Y_{0:t-1} \mid X_t}^*$ is uniquely determined by the integral equation
$$
    \int_{\cX \x \cW^t}
    g(y,F_t(y))
    (P_* \x R^t)(\rd y)
    =
    \int_{\cX}
    \Big(
        \int_{\cX \x \cW^t}
        g(y,x)
        \sP_{Y_{0:t-1} \mid X_t}^*(\rd y \mid x)
    \Big)
    P_*(\rd x),
$$
which must be satisfied
for Borel measurable functions $g: (\cX \x \cW^t)\x \cX \to \mR$ and is closely related to Bayes formula. Here, use is made of the property that the random element $F_t(\eta)$ in Definition \ref{def:nom_post} has the nominal invariant state distribution  $P_*$.

\begin{lemma}
\label{lem:rel_ent_balance} Suppose the initial state distribution of the system (\ref{system}) satisfies  $P_0 \ll P_*$, and the noise $W$ is admissible in the sense of Definition~\ref{def:good}. Then for any $t>0$,  the conditional relative entropy $E_{0,t}$, defined by
(\ref{Est}), is representable as
\begin{equation}
\label{balance}
    E_{0,t}
     =
    \bD_0(P_t \| P_*)
    -
    \bD_0(P_0 \| P_*)
      +
    \bD(
        \sP_{Y_{0:t-1} \mid X_t}
        \|
        \sP_{Y_{0:t-1} \mid X_t}^*
    ).
\end{equation}
Here, $\bD_0$ is the relative entropy functional (\ref{bD0}), and $\sP_{Y_{0:t-1} \mid X_t}^*$  is the nominal posterior distribution of the state-noise sequence $Y_{0:t-1}$ from Definition \ref{def:nom_post}.
\end{lemma}
\begin{proof}
The factorization (\ref{QPR}) (see also the chain rule for the relative entropy
\cite[Lemma 5.3.1 on p. 94]{Gray_2009}) implies that
\begin{align}
\nonumber
    \bD_0(Q_{0,t}\| P_* \x R^t)
    & =
    \bE\ln (\varpi_0(X_0) \varphi_{0,t}(W_{0:t-1} \mid X_0))\\
\nonumber
    & =
    \bE\ln \varpi_0(X_0)
    +
    \bE\ln \varphi_{0,t}(W_{0:t-1} \mid X_0)\\
\label{chain1}
    & =
    \bD_0(P_0 \| P_*)
    +
    E_{0,t},
\end{align}
where the expectation is taken over the actual probability distribution $Q_{0,t}$ of the state-noise sequence $Y_{0:t-1}$. Here,
$\varphi_{0,t}$ is the conditional PDF of $W_{0:t-1}$ given $X_0$ from (\ref{Est}), and $\varpi_0$ is the PDF (\ref{varpi}) of the initial state distribution $P_0$ with respect to the nominal invariant state distribution $P_*$. Furthermore,  since $X_t = F_t(Y_{0:t-1})$ depends in a deterministic way on $Y_{0:t-1}$ in view of (\ref{FG}), so that the conditional distribution $\sP_{X_t \mid Y_{0:t-1}}(\cdot \mid y)$ is an atomic  probability measure \cite[p. 46]{Schilling_2005}  concentrated on the singleton $\{F_t(y)\}$ for any $y \in \cX\x \cW^t$ regardless of the probability distribution of $Y_{0:t-1}$, then the augmentation of $Y_{0:t-1}$ by $X_t$ does not change the relative entropy in (\ref{chain1}). More precisely, by using a $P_*\x R^t$-distributed random element $\eta$ from Definition~\ref{def:nom_post} and applying the relative entropy chain rule again, it follows that
 \begin{align}
\nonumber
    \bD_0(\sP_{Y_{0:t-1}, X_t} \| \sP_{\eta, F_t(\eta)})
    & =
    \bD_0(\sP_{Y_{0:t-1}} \| \sP_{\eta}) +
    \bD(\sP_{X_t \mid Y_{0:t-1}} \| \sP_{F_t(\eta) \mid \eta})\\
 \label{DDD}
    & =
    \bD_0(Q_{0,t}\| P_* \x R^t).
 \end{align}
Here, $\bD(\sP_{X_t \mid Y_{0:t-1}} \| \sP_{F_t(\eta) \mid \eta}) =0$ because the conditional probability distributions  $\sP_{X_t \mid Y_{0:t-1}}(\cdot \mid y)$ and $\sP_{F_t(\eta) \mid \eta}(\cdot\mid y)$ are identical to each other as  discussed above. Now, application of the relative entropy chain rule to the left-hand side of (\ref{DDD}) in the opposite time direction, with $Y_{0:t-1}$ being conditioned on $X_t$, yields
 \begin{align}
\nonumber
    \bD_0(\sP_{Y_{0:t-1}, X_t} \| \sP_{\eta, F_t(\eta)})
    & =
    \bD_0(\sP_{X_t} \| \sP_{F_t(\eta)}) +
    \bD(\sP_{Y_{0:t-1}\mid X_t} \| \sP_{\eta\mid F_t(\eta)})\\
 \label{DDDD}
    & =
    \bD_0(P_t \| P_*) +
    \bD(\sP_{Y_{0:t-1}\mid X_t} \| \sP_{Y_{0:t-1}\mid X_t}^*).
 \end{align}
 Here, use is made of Definition~\ref{def:nom_post} of the nominal posterior distribution $\sP_{Y_{0:t-1}\mid X_t}^*$ and the property that $F_t(\eta)$ is $P_*$-distributed. Also, the absolute continuity $P_t \ll P_*$ is ensured by the assumption that $P_0 \ll P_*$ and the admissibility of the noise in the sense of Definition~\ref{def:good}. By a straightforward comparison of (\ref{chain1})--(\ref{DDDD}), it follows that
$$    \bD_0(Q_{0,t}\| P_* \x R^t)
    =
   \bD_0(P_0 \| P_*)
    +
    E_{0,t}
    =
    \bD_0(P_t \| P_*)
    +
    \bD(
        \sP_{
            Y_{0:t-1} \mid X_t
        }
        \|
        \sP_{
            Y_{0:t-1} \mid X_t
        }^*
    ),
$$
where the second equality is equivalent to the representation (\ref{balance}), and the proof of the lemma is completed.
\hfill\end{proof}

The conditional relative entropy $E_{0,t}$ in (\ref{Est}) can be interpreted
as the \emph{supply} which  the noise player has to deliver to the
system over the time interval $[0,t)$ in order to make the state distribution $P_t$ of the system deviate from the nominal invariant state
distribution $P_*$. In view of the relative entropy balance equation (\ref{balance}), only part of this
 ``expenditure'', namely, $\bD_0(P_t \| P_*) - \bD_0(P_0 \|
P_*)$, contributes directly  to achieving this goal. The rest of the \emph{conditional relative entropy supply}  $E_{0,t}$ is
``dissipated'' into     $\bD(\sP_{Y_{0:t-1} \mid X_t}\| \sP_{Y_{0:t-1} \mid X_t}^*)$ which quantifies the amount by which the actual conditional
probability  distribution of the state-noise sequence $Y_{0:t-1}$
given $X_t$ can be distinguished from the nominal posterior
distribution $\sP_{Y_{0:t-1} \mid X_t}^*$. This dissipation is caused by an irreversible
loss of information contained in  the state-noise sequence
$Y_{0:t-1}$, only a fraction of which is able to be
encoded in the terminal state $X_t=F_t(Y_{0:t-1})$ of the system in a bijective way.  Omitting the term $\bD(\sP_{Y_{0:t-1} \mid X_t}\| \sP_{Y_{0:t-1} \mid X_t}^*)\> 0$,  the equality
(\ref{balance}) implies that
\begin{equation}
\label{dissipation}
    \bD_0(P_t \| P_*)
    \<
    \bD_0(P_0 \| P_*)
    +
    E_{0,t}.
\end{equation}
By analogy with deterministic dissipativity theory
\cite[pp.~327, 348]{Willems_1972}, the relation (\ref{dissipation}) describes a \emph{relative entropy dissipation inequality}. Accordingly, the \emph{state relative
entropy} $\bD_0(P_t\| P_*)$, which quantifies the deviation of the  actual state distribution $P_t$ from the nominal
invariant state distribution $P_*$, plays the role of a \emph{storage
function} at time $t$. Note, however, that in the stochastic  setting under consideration, these entropy theoretic functionals do not inherit all the properties of the corresponding concepts for deterministic
dissipative systems. For example, unlike the deterministic integral
supply which, as a function of an interval of time, is additive  \cite[p. 23]{Schilling_2005} with respect to the union of disjoint time intervals,  the conditional relative entropy supply (\ref{Est}) is, in general,  \emph{superadditive} as described below.

\begin{lemma}
\label{lem:superadd} For any $0<s<t$, the conditional
relative entropy supply $E_{0,t}$ over the time interval $[0,t)$, defined by (\ref{Est}), is not less
than the sum of the supplies over the constituent
subintervals $[0,s)$ and $[s,t)$:
\begin{equation}
\label{superadd}
    E_{0,t}
    \>
    E_{0,s}
    +
    E_{s,t}.
\end{equation}
The inequality (\ref{superadd}) becomes an equality if and only if three random
elements $Y_{0:s-1}$, $X_s$, $W_{s:t-1}$ form a Markov chain.
\end{lemma}
\begin{proof}
The chain rule for joint PDFs with respect to product measures allows the conditional
PDF $\varphi_{0,t}$ in (\ref{Est}) to be factorized as
\begin{align}
\nonumber
    \varphi_{0,t}(w_0, \ldots, w_{t-1}\mid x_0)
    = &
    \frac
    {\sP_{ W_{0:t-1} \mid X_0}
    (\rd w_0\x  \ldots\x\rd w_{t-1}\mid x_0) }
    {R(\rd w_0)\x \ldots \x R(\rd w_{t-1})}\\
\nonumber
    = &
    \frac
    {\sP_{ W_{0:s-1} \mid X_0}
    (\rd w_0\x  \ldots\x\rd w_{s-1}\mid x_0) }
    {R(\rd w_0)\x \ldots \x R(\rd w_{s-1})}\\
\nonumber
    &  \x
    \frac
    {\sP_{ W_{s:t-1} \mid Y_{0:s-1}}(\rd w_s\x  \ldots\x\rd w_{t-1}\mid x_0, w_0, \ldots, w_{s-1}) }
    {R(\rd w_s)\x \ldots \x R(\rd w_{t-1})}\\
\label{PPP}
    = &
    \varphi_{0,s}(w_0, \ldots, w_{s-1}\mid x_0)
    \psi_{s,t}(w_s, \ldots, w_{t-1}\mid x_0, w_0, \ldots, w_{s-1})
\end{align}
for all $x_0 \in \cX$ and $w_0, \ldots, w_{t-1} \in \cW$.
Here, $\psi_{s,t}: \cW^{t-s}\x \cX \x \cW^s\to \mR_+$ is the conditional PDF of the noise sequence $W_{s:t-1}$, given the state-noise sequence
$Y_{0:s-1}$, with respect to the reference measure $R^{t-s}$:
\begin{equation}
\label{psist}
    \psi_{s,t}(w_s, \ldots, w_{t-1}\mid x_0, w_0, \ldots, w_{s-1})
    :=
    \frac
    {\sP_{ W_{s:t-1} \mid Y_{0:s-1}}(\rd w_s\x  \ldots\x\rd w_{t-1}\mid x_0, w_0, \ldots, w_{s-1}) }
    {R(\rd w_s)\x \ldots \x R(\rd w_{t-1})}.
 \end{equation}
 Therefore, substitution of (\ref{PPP}) and (\ref{psist}) into the definition  (\ref{Est}) of the conditional relative entropy $E_{0,t}$ yields
\begin{eqnarray}
\nonumber
    E_{0,t}
    & = &
    \bE\ln\varphi_{0,t}(W_{0:t-1}\mid X_0)\\
\nonumber
    & = &
    \bE\ln(\varphi_{0,s}(W_{0:s-1}\mid X_0)\psi_{s,t}(W_{s:t-1}\mid Y_{0:s-1}))\\
\nonumber
    & = &
    \bE\ln\varphi_{0,s}(W_{0:s-1}\mid X_0)
    +
    \bE\ln\psi_{s,t}(W_{s:t-1}\mid Y_{0:s-1})\\
\label{superadd1}
& = &
    E_{0,s}
    +
    \bD(\sP_{W_{s:t-1}\mid Y_{0:s-1}} \mid R^{t-s}).
\end{eqnarray}
The inequality (\ref{superadd}), which describes the superadditivity of the conditional relative entropy supply,  can now be obtained by combining (\ref{superadd1}) with
\begin{equation}
\label{Jensen}
    \bD
    (
        \sP_{
            W_{s:t-1}
            \mid
            Y_{0:s-1}
        }
        \|
        R^{t-s}
    )
    \>
    \bD
    (
        \sP_{
            W_{s:t-1}
            \mid
            X_s
        }
        \|
        R^{t-s}
    )
=
    E_{s,t}.
\end{equation}
The last inequality follows from the property that $X_s = F_s(Y_{0:s-1})$ depends in a deterministic way on $Y_{0:s-1}$ through a  Borel measurable map, whereby the conditioning on $Y_{0:s-1}$ is finer than that on $X_s$.  For a rigorous proof of the inequality in (\ref{Jensen}), we consider the conditional probability distribution of $Y_{0:s-1}$ given $X_s$:
\begin{equation}
\label{theta}
    \theta_s(B\mid x)
    :=
    \sP_{ Y_{0:s-1} \mid X_s}(B\mid x),
    \qquad
    B\in \fX\x \fW^s,\
        x \in \cX.
\end{equation}
Recall that such posterior distributions of the state-noise sequences were used in Lemma~\ref{lem:rel_ent_balance}.
In view of (\ref{FG}),
for any given $x\in \cX$, the probability measure  $\theta_s(\cdot \mid x)$ on $(\cX\x \cW^s, \fX \x \fW^s)$ is concentrated on the pre-image $F_s^{-1}(x)= \{y \in \cX\x \cW^s:\ F_s(y)=x\}$
of the point $x$ under the $s$ step state transition map $F_s$ in the sense that
$\theta_s(F_s^{-1}(x) \mid x) = 1$.
Then the conditional PDF $\varphi_{s,t}$ from (\ref{Est}) is representable as an appropriate average of the conditional PDF (\ref{psist}) over the probability measure (\ref{theta}) as
\begin{equation}
\label{phipsitheta}
    \varphi_{s,t}(w \mid x)
    =
    \bE(\psi_{s,t}(w \mid Y_{0:s-1}) \mid X_s = x)
    =
    \int_{F_s^{-1}(x)}
    \psi_{s,t}(w\mid y)
    \theta_s(\rd y\mid x)
\end{equation}
for all $w\in \cW^{t-s}$ and $x\in \cX$.
Since the function $\bL$ in (\ref{bL}) is
strictly convex, then (\ref{phipsitheta}) and Jensen's
inequality \cite{Shiryaev} imply that
\begin{equation}
\label{Jensen1}
    \bL(\varphi_{s,t}(w\mid x))
    \<
    \bE(\bL(\psi_{s,t}(w \mid Y_{0:s-1})) \mid X_s = x)
    =
    \int_{F_s^{-1}(x)}
    \bL(\psi_{s,t}(w\mid y))
    \theta_s(\rd y\mid x).
\end{equation}
Moreover, the inequality in (\ref{Jensen1}) becomes an equality if and only if $\psi_{s,t}(w\mid y) = \varphi_{s,t}(w\mid x)$ holds for $\theta_s(\cdot \mid x)$-almost all $y \in F_s^{-1}(x)$.
Hence, the conditional relative entropy
supply $E_{s,t}$ in (\ref{Est})
 satisfies
\begin{align}
\nonumber
    E_{s,t}
& =
    \int_{\cX}
    \Big(
        \int_{\cW^{t-s}}
        \bL(\varphi_{s,t}(w\mid x))
        R^{t-s}(\rd w)
    \Big)
    P_s(\rd x)\\
\nonumber
& \<
    \int_{\cX}
    \Big(
     \int_{\cW^{t-s}}
     \Big(
     \int_{F_s^{-1}(x)}
    \bL(\psi_{s,t}(w\mid y))
    \theta_s(\rd y\mid x)
    \Big)
    R^{t-s}(\rd w)
    \Big)
    P_s(\rd x)\\
\nonumber
& =
    \int_{\cX}
    \Big(
     \int_{F_s^{-1}(x)}
     \Big(
     \int_{\cW^{t-s}}
    \bL(\psi_{s,t}(w\mid y))
    R^{t-s}(\rd w)
    \Big)
    \theta_s(\rd y\mid x)
    \Big)
    P_s(\rd x)\\
\nonumber
& =
     \int_{\cX \x \cW^s}
     \Big(
     \int_{\cW^{t-s}}
    \bL(\psi_{s,t}(w\mid y))
    R^{t-s}(\rd w)
    \Big)
    Q_{0,s}(\rd y)\\
\label{Est1}
& =
    \bD
    (
        \sP_{
            W_{s:t-1}
            \mid
            Y_{0:s-1}
        }
        \|
        R^{t-s}
    ),
\end{align}
which establishes the inequality in (\ref{Jensen}) and completes the proof  of (\ref{superadd}). Turning
to the second part of the lemma, note that
(\ref{superadd}) becomes an equality if and only if  the
inequality in (\ref{Est1}) becomes an equality. By the strict convexity of the function $\bL$ from (\ref{bL}), it follows from (\ref{phipsitheta}) and (\ref{Jensen1})
that the inequality in (\ref{Est1}) becomes an equality if and only
if
\begin{equation}
\label{Markov}
    \psi_{s,t}(w \mid y) = \varphi_{s,t}(w\mid F_s(y))
    \qquad
    {\rm for}\ Q_{0,t}{\rm-almost\ all}\ (y,w) \in (\cX\x \cW^s)\x \cW^{t-s}.
\end{equation}
In view of (\ref{Est}) and (\ref{psist}), the relation
(\ref{Markov}) holds if and only if $\sP_{ W_{s:t-1} \mid
Y_{0:s-1}}$ depends on $Y_{0:s-1}$ only through $X_s =
F_s(Y_{0:s-1})$, which is equivalent to the condition  that the  three random elements $Y_{0:s-1}$, $X_s$,
$W_{s:t-1}$ form a Markov chain.
\hfill\end{proof}

As can be seen from the above proof, Lemma~\ref{lem:superadd} is closely
related to the data processing inequality and convexity of the
relative entropy functional with respect to each of its arguments \cite{CT_2006,Gray_2009}. The second assertion of
Lemma~\ref{lem:superadd}  shows that the conditional relative entropy supply $E_{s,t}$
is \emph{additive} as a function of the time
interval $[s,t)$ (that is,
$
    E_{s,u} = E_{s,t} + E_{t,u}
$ for all $0\< s< t < u$)
 if only if the noise sequence $W$ is Markov with respect to the state sequence
$X$ of the system. The Markov property means that, for \emph{every}  time $k\> 0$, the probability measures $\sP_{W_k  \mid Y_{0:k-1}}(\cdot \mid y)$ and $\sP_{W_k \mid X_k}(\cdot \mid F_k(y))$ on $(\cW, \fW)$ are equal to each other for $Q_{0,k}$-almost all values $y\in \cX \x \cW^k$ of the state-noise sequence $Y_{0:k-1}$. It turns out that the Markov property  is important for noise strategies  to be economical in the sense of the conditional relative entropy supply.

\section{Markov noise strategies decrease the conditional relative entropy supply}
\label{sec:markovization}

We will now introduce a specific change of measure which leads to the Markov property of the noise with respect to the state of the system. More precisely, for any $s> 0$, consider an operator $M_s$  which, for any $t> s$,   maps
the probability distribution $Q_{0,t}$ of the state-noise
sequence $Y_{0:t-1}$, associated with an admissible noise $W$,   to another probability
distribution $\wh{Q}_{0,t} := M_s(Q_{0,t})$ on the same
measurable space $(\cX\x \cW^{t-s}, \fX\x \fW^{t-s})$ as
\begin{equation}
\label{M}
    \wh{Q}_{0,t}(\rd y \x \rd w)
    =
    \varphi_{s,t}(w \mid F_s(y))
    Q_{0,s}(\rd y)
    R^{t-s}(\rd w),
    \qquad
    y \in \cX\x \cW^s,\
    w \in \cW^{t-s}.
\end{equation}
Here,
$\varphi_{s,t}$ is the conditional PDF of the noise sequence $W_{s:t-1}$ given $X_s$ from (\ref{Est}), associated with the original distribution $Q_{0,t}$, and $F_s$ is the $s$ step state transition map from (\ref{FG}) and (\ref{Fnext}). In order to clarify the meaning of (\ref{M}), note that
\begin{equation}
\label{QQQ}
    Q_{0,t}(\rd y \x \rd w)
    =
    \psi_{s,t}(w \mid y)
    Q_{0,s}(\rd y)
    R^{t-s}(\rd w),
\end{equation}
in view of the factorization   (\ref{PPP}) and the definition of the conditional PDF $\psi_{s,t}$ in (\ref{psist}). Direct comparison of (\ref{M}) with (\ref{QQQ}) shows that the
action of the operator $M_s$ on $Q_{0,t}$ leads to the
Markov property of the state-noise sequence $Y_{0:t-1}$ with respect to
the intermediate state $X_s$ by replacing the left-hand side $\psi_{s,t}(w \mid y)$ of (\ref{Markov}) with its right-hand side $\varphi_{s,t}(w \mid F_s(y))$. Therefore, an equivalent representation of $M_s$ in terms of the conditional PDFs from (\ref{Est}) is
\begin{equation}
\label{FFF}
    \wh{\varphi}_{0,t}(w_0, \ldots, w_{t-1}\mid x_0)
    :=
    \varphi_{0,s}(w_0, \ldots, w_{s-1}\mid x_0)
    \varphi_{s,t}(w_s, \ldots, w_{t-1}\mid F_s(x_0, w_0, \ldots, w_{s-1}))
\end{equation}
for all $x_0\in \cX$ and $w_0, \ldots, w_{t-1}\in \cW$,
where $\wh{\varphi}_{0,t}$ corresponds to $\wh{Q}_{0,t}$ in (\ref{M}), whilst $\varphi_{0,s}$ and $\varphi_{s,t}$ are associated with $Q_{0,s}$ and $Q_{s,t}$.
Under the new
measure $\wh{Q}_{0,t}$, the random elements $Y_{0:s-1}$, $X_s$,
$W_{s:t-1}$ form a Markov chain. The operator
$M_s$ is idempotent (that is, $M_s^2:= M_s\circ M_s = M_s$), since those (and only
those) probability distributions $Q_{0,t}$, which are already Markov with
respect to $X_s$, are invariant under $M_s$.

\begin{lemma}
\label{lem:markovization} For any $0<s<t$, the
operator $M_s$, acting on the probability distribution $Q_{0,t}$ in
(\ref{QQQ}) as described by (\ref{M}), leaves the probability
distributions $Q_{0,s}$ and $Q_{s,t}$ and the state distributions
$P_0, \ldots, P_t$ unchanged. The conditional relative entropy
supply $\wh{E}_{0,t}$ on the time interval $[0,t)$, associated
with the new measure $\wh{Q}_{0,t}$, satisfies
\begin{equation}
\label{EEEE}
    \wh{E}_{0,t}
    =
    E_{0,s} + E_{s,t}
    \<
    E_{0,t}.
\end{equation}
The inequality in (\ref{EEEE}) becomes an equality if and only if
the measure $Q_{0,t}$ is invariant under $M_s$, that is, if and only
if the three random elements $Y_{0:s-1}$, $X_s$, $W_{s:t-1}$ form a
Markov chain.
\end{lemma}

\begin{proof}
Throughout the proof, the probability distributions and other quantities associated with
the new measure $\wh{Q}_{0,t}$ will be marked by the ``hat''
symbol. The property that the operator $M_s$ preserves
the probability distribution of $Y_{0:s-1}$,
\begin{equation}
\label{QQ0s}
    \wh{Q}_{0,s} = Q_{0,s},
\end{equation}
is verified by using an appropriate marginal distribution obtained from $\wh{Q}_{0,t}$
via integrating both sides of  (\ref{M}) over $w\in \cW^{t-s}$, since $\int_{\cW^{t-s}}\varphi_{s,t}(w\mid x) R^{t-s}(\rd w)=1$  for any $x \in \cX$. Hence, the conditional relative entropy
supply $E_{0,s}$ over the time interval $[0,s)$, which is completely specified by $Q_{0,s}$, remains unaffected:
\begin{equation}
\label{EE0s}
    \wh{E}_{0,s} = E_{0,s}.
\end{equation}
Furthermore, (\ref{QQ0s}) implies that the state distributions $P_0, \ldots, P_s$ of
the system up until time $s$  are also preserved:
\begin{equation}
\label{PP0s}
    \wh{P}_k = P_k,
    \qquad
    k =0, \ldots, s.
\end{equation}
It follows from  (\ref{M}) that $M_s$ also
preserves the conditional probability distribution $\sP_{ W_{s:t-1}\mid
X_s}$. Hence, the conditional PDF $\varphi_{s,t}$ from (\ref{Est}) is also preserved: $
    \wh{\varphi}_{s,t} = \varphi_{s,t}
$. This property, combined  with the equality $\wh{P}_s = P_s$ from
(\ref{PP0s}), yields
\begin{align}
\nonumber
    \wh{Q}_{s,t}(\rd x \x \rd w)
& =
    \wh{\varphi}_{s,t}(w\mid x)
    \wh{P}_s(\rd x)
    R^{t-s}(\rd w)\\
\nonumber
& =
    \varphi_{s,t}(w\mid x)
    P_s(\rd x)
    R^{t-s}(\rd w)\\
\label{QQst}
    & =
    Q_{s,t}(\rd x \x \rd w),
    \qquad
    x \in \cX,\
    w \in \cW^{t-s}.
\end{align}
Therefore, the conditional
relative entropy supply over the time interval $[s,t)$ is also
invariant under the action of $M_s$:
\begin{equation}
\label{EEst}
    \wh{E}_{s,t} = E_{s,t}.
\end{equation}
Furthermore, (\ref{QQst}) implies the invariance of the corresponding
state distributions $P_s, \ldots, P_t$ of the system:
\begin{equation}
\label{PPst}
    \wh{P}_k = P_k,
    \qquad
    k=s, \ldots, t.
\end{equation}
Therefore, in view of (\ref{PP0s}) and (\ref{PPst}), all the state distributions $P_0,
\ldots, P_t$ of the system on the time interval $[0,t]$ are
invariant under $M_s$. Since $Y_{0:s-1}$, $X_s$, $W_{s:t-1}$ form
a Markov chain with respect to the new measure $\wh{Q}_{0,t}$, then
a combination of Lemma~\ref{lem:superadd} with (\ref{EE0s}) and
(\ref{EEst}) yields
\begin{equation}
\label{sixE}
    \wh{E}_{0,t}
    =
    \wh{E}_{0,s} + \wh{E}_{s,t}
    =
    E_{0,s} + E_{s,t}
    \<
    E_{0,t},
\end{equation}
which proves (\ref{EEEE}). Using  the second assertion of
Lemma~\ref{lem:superadd}, it follows that the inequality in (\ref{sixE}) is an
equality if and only if $Y_{0:s-1}$, $X_s$, $W_{s:t-1}$ form a
Markov chain under the original measure $Q_{0,t}$. Therefore, to prove the last assertion of Lemma~\ref{lem:markovization}, it remains only to recall the equivalence between the Markov property of the probability measure $Q_{0,t}$ and its
invariance $M_s(Q_{0,t}) = Q_{0,t}$ under the operator $M_s$ defined in (\ref{M}).
\hfill\end{proof}

Using Lemma~\ref{lem:markovization}, it follows that the application of the
operator $M_s$ \emph{strictly decreases} the conditional relative entropy supply
over the time interval $[0,t)$, thereby leading to a more economical
strategy for the noise player,
 unless the original noise strategy is already Markov with respect to $X_s$. From  (\ref{M}), it follows that
the operators $M_s$ and $M_u$ commute  for any  times $s<u$, and their composition
$
    M_s\circ M_u
    =
    M_u\circ M_s
$
leads to the Markov property of the noise $W$ with respect to the
states $X_s$ and $X_u$. More generally, the operator
\begin{equation}
\label{MMM}
    M_{1,\ldots, t-1}
    :=
    M_1
        \circ
        \ldots
        \circ
    M_{t-1}
\end{equation}
leads to the Markov property of the noise $W$ with respect to the intermediate
states $X_1, \ldots, X_{t-1}$.
 The resulting probability distribution $\wh{Q}_{0,t}:=
M_{1,\ldots, t-1}(Q_{0,t})$ of the state-noise sequence $Y_{0:t-1}$, whose conditional PDF
is given by
$$
    \wh{\varphi}_{0,t}(w_0, \ldots, w_{t-1}\mid x_0)
    :=
    \prod_{s=0}^{t-1}
    \varphi_{s,s+1}(w_s\mid F_s(x_0, w_0, \ldots, w_{s-1})),
$$
similarly to (\ref{FFF}),
inherits the distributions $Q_{k,k+1}=\sP_{X_k, W_k}$ from
$Q_{0,t}$ for all $k=0, \ldots, t-1$.  Furthermore, the conditional relative
 entropy supply $\wh{E}_{0,t}$, associated with  $\wh{Q}_{0,t}$, is additive on the time interval $[0,t)$ in the sense
 of the equalities
 $$
    \wh{E}_{s,u}
    =
    \sum_{k=s}^{u-1}
    E_{k,k+1}
    \< E_{s,u},
    \qquad
    0\< s < u \<  t.
 $$
The fact that  the operator $M_{1, \ldots, t-1}: Q_{0,t}\mapsto \wh{Q}_{0,t}$ in (\ref{MMM}) decreases the conditional relative entropy supply, while preserving  the state distributions of the system, implies that a non-Markov noise strategy, which drives the system along a given sequence of state distributions $P_0, \ldots, P_t$, can always be made more economical by replacing $Q_{0,t}$ with the Markov strategy $\wh{Q}_{0,t}$.

\section{Bellman equation for the minimum required conditional relative entropy supply}
\label{sec:Bellman}

 Consider the problem of minimizing
the conditional relative entropy supply $E_{0,t}$ in (\ref{Est})  required to drive
 the system (\ref{system}) from a given initial state distribution $\Phi$
to a given terminal state distribution $\Psi$ over a time interval of
specified length $t
> 0$:
\begin{equation}
\label{Jt}
    J_t(\Phi,\Psi)
    :=
    \inf
    \big\{
        E_{0,t}:\
        P_0=\Phi,\
        \sP_{W_{0:t-1}\mid X_0}\ll R^t,\
        P_t =\Psi
    \big\}.
\end{equation}
Here, both probability measures $\Phi$ and $\Psi$ on $(\cX,\fX)$ are assumed to be
absolutely continuous with respect to the nominal
invariant state distribution $P_*$, and the infimum is taken over those admissible noise strategies $\sP_{W_{0:t-1}\mid X_0}$ in the sense of Definition~\ref{def:good}, under which the state distribution of the system evolves from $P_0=\Phi$ to $P_t = \Psi$.
Variational problems like (\ref{Jt}), which involve entropy and probabilistic boundary conditions, are known as Schr\"{o}dinger bridge problems \cite{Beghi_1996,PT_2010} and are usually treated in the context of reciprocal processes, that is, Markov random fields on the time axis  \cite{Jamison_1974}; see also \cite{Beghi_1994,Blaq_1992,DaiPra_1991,Mikami_1990,VP_2010a} for continuous time formulations.

If the system is initialized at a $P_*$-distributed state $X_0$, then application of a nominal noise with $\sP_{ W \mid X_0 } = R^{\infty}$ (so that $E_{0,t} = 0$) leaves the state distribution unchanged, and hence,
\begin{equation}
\label{rest}
    J_t(P_*,P_*) = 0
\end{equation}
for any time horizon $t>0$. However, if $\Psi\ne P_*$, then $J_t(P_*, \Psi)$  is positive and quantifies the cost for the noise player to drive the system from $P_*$ to $\Psi$ in time $t$. The larger $J_t(P_*, \Psi)$ is, the more robust the system is with respect to the uncertain noise. 
Minimization  of $E_{0,t}$ on the right-hand side of
the dissipation inequality (\ref{dissipation}) under the
constraints $P_0 = \Phi$ and $P_t = \Psi$ yields a lower bound
\begin{equation}
\label{lower_bound}
    J_t(\Phi, \Psi)
    \>
    \max
    (
        \bD_0(\Psi \| P_*)
        -
        \bD_0(\Phi \| P_*),\,
        0
    ),
\end{equation}
which also clarifies the role of the assumptions $\Phi\ll P_*$ and $\Psi\ll P_*$ for the well-posedness of the problem (\ref{Jt}). However, these absolute continuity conditions are, in general, not enough to guarantee finiteness of the quantity $J_t(\Phi,\Psi)$ since the discrete-time system (\ref{system}) may lack reachability with respect to the noise over short time intervals.

\begin{definition}
A terminal state distribution $\Psi\ll P_*$ is said to be
reachable from an  initial state distribution $\Phi\ll P_*$   in time $t>0$ if the minimum required conditional relative entropy supply
$J_t(\Phi, \Psi)$ in (\ref{Jt}) is finite.
\end{definition}

The following theorem shows that the additivity of the conditional relative entropy supply for Markov noise strategies, established in Lemma~\ref{lem:superadd}, plays an important role in determining the minimum required supply in (\ref{Jt}).

\begin{theorem}
\label{th:sum} For any time horizon $t>0$, intermediate time $0<s<t$
and initial and terminal state distributions $\Phi$ and $\Psi$, the
minimum required conditional relative entropy supply (\ref{Jt}) satisfies
\begin{equation}
\label{sum}
    J_t(\Phi, \Psi)
    =
    \inf_{\Theta}
    \left(
        J_s(\Phi,\Theta)
        +
        J_{t-s}(\Theta, \Psi)
    \right),
\end{equation}
where the infimum is taken over all intermediate state distributions
$\Theta$ reachable from $\Phi$ in time $s$ and for which
$\Psi$ is reachable from $\Theta$ in time $t-s$. Furthermore, if the infimum in (\ref{Jt}) is
achieved, then every optimal noise strategy is Markov with respect
to the state of the system.
\end{theorem}

\begin{proof}
By using an
intermediate state distribution $\Theta$, it follows that the infimum in (\ref{Jt}) can be represented as
\begin{equation}
\label{auxiliary}
    J_t(\Phi, \Psi)
    =
    \inf_{\Theta}
    J_{s,t}(\Phi,\Theta,\Psi),
\end{equation}
where
\begin{equation}
\label{Jst}
    J_{s,t}(\Phi,\Theta,\Psi)
    :=
    \inf
    \big\{
        E_{0,t}:\
        \sP_{W_{0:t-1}\mid X_0}\ll R^t,\
        P_0=\Phi,\
        P_s=\Theta,\
        P_t =\Psi
    \big\}
\end{equation}
involves an additional constraint $P_s = \Theta$.
Application of the superadditivity (\ref{superadd}) of the conditional
relative entropy supply to
(\ref{Jst}) yields
\begin{equation}
\label{three}
    J_{s,t}(\Phi,\Theta,\Psi)
        \>
    J_s(\Phi, \Theta)
    +
    J_{t-s}(\Theta, \Psi).
\end{equation}
We will now prove that the inequality (\ref{three}) is, in fact, an
equality from which (\ref{sum})  follows immediately in view of
(\ref{auxiliary}). Suppose the probability distributions $Q_{0,s}$ and $Q_{s,t}$ are associated with an admissible noise on the subintervals $[0,s)$ and $[s,t)$ satisfying $P_0 =
\Phi$, $P_s = \Theta$, $P_t = \Psi$. Note that  $Q_{0,s}$ and
$Q_{s,t}$ are compatible since they ascribe to $X_s$ the same
probability distribution $\Theta$. Hence, there exists a probability
distribution $\wh{Q}_{0,t}$ which is Markov with respect to the
intermediate state $X_s$ and leads to the marginal
distributions $Q_{0,s}$ and $Q_{s,t}$ described above. The corresponding conditional PDF $\wh{\varphi}_{0,t}: \cW^t\x \cX \to \mR_+$ from (\ref{Est})  is expressed in terms of $\varphi_{0,s}$ and $\varphi_{s,t}$, associated with $Q_{0,s}$ and $Q_{s,t}$, as described by (\ref{FFF}).
In addition to $P_s=\Theta$, the measure $\wh{Q}_{0,t}$ also satisfies the boundary conditions $P_0=\Phi$ and $P_t=\Psi$ for the state distribution.  By Lemma~\ref{lem:superadd},  the Markov property of $\wh{Q}_{0,t}$ implies that the conditional relative entropy supply satisfies
\begin{equation}
\label{EEE}
    \wh{E}_{0,t} = E_{0,s} + E_{s,t}.
\end{equation}
For any $\eps >0$, each of the measures $Q_{0,s}$
and $Q_{s,t}$ can be chosen so that the corresponding conditional
relative entropy supply is $\eps$-close to its
minimal value in (\ref{Jt}):
\begin{equation}
\label{EE}
   E_{0,s} \< J_s(\Phi,\Theta) + \eps,
    \qquad
    E_{s,t} \< J_{t-s}(\Theta,\Psi) + \eps.
\end{equation}
By combining (\ref{EEE}) and (\ref{EE}), it follows that
$$
    \wh{E}_{0,t}
    \<
    J_s(\Phi,\Theta)
    +
    J_{t-s}(\Theta,\Psi) + 2\eps.
$$
Therefore, by combining the suboptimal  noise strategies
$Q_{0,s}$ and $Q_{s,t}$ into the Markov strategy $\wh{Q}_{0,t}$ as described by
(\ref{FFF}), the total conditional relative entropy supply
$\wh{E}_{0,t}$ can be made arbitrarily close to the right-hand side of
the inequality (\ref{three}). This implies that  (\ref{three}) holds
as an equality, thus proving (\ref{sum}) in view of (\ref{auxiliary}). We now proceed to the proof of the second assertion of the theorem which assumes that the infimum in (\ref{Jt}) is
achievable.  Let $Q_{0,t}$ be an optimal noise strategy which leads to
the minimum conditional relative entropy supply
$E_{0,t}=J_t(\Phi,\Psi)$. Suppose $Q_{0,t}$ is not Markov with
respect to the state signal $X$ of the system on the time interval
$[0,t)$. Then application of the operator
(\ref{MMM}) generates a different measure $\wh{Q}_{0,t}:= M_{1,
\ldots, t-1}(Q_{0,t}) \ne Q_{0,t}$. By
Lemma~\ref{lem:markovization}, $\wh{Q}_{0,t}$ satisfies the
boundary conditions  $P_0=\Phi$ and $P_t = \Psi$ for the state
distribution   and delivers a smaller conditional relative entropy
supply $\wh{E}_{0,t} < E_{0,t}$. The latter, however,
contradicts the optimality of $Q_{0,t}$.  This contradiction
implies the Markov property of $Q_{0,t}$.
\hfill\end{proof}

Special cases of Theorem~\ref{th:sum}, which are obtained by letting $s=1$ or $s=t-1$ in (\ref{sum}),  lead to a dynamic programming Bellman equation \cite[pp. 319--320]{LS_1995} for the minimum conditional relative entropy supply  in
(\ref{Jt}):
\begin{equation}
\label{Bellman}
    J_{t+1}(\Phi, \Psi)
     =
    \inf_{\Theta}
    (
        J_1(\Phi,\Theta)
        +
        J_t(\Theta, \Psi)
    )
     =
    \inf_{\Theta}
    (
        J_t(\Phi,\Theta)
        +
        J_1(\Theta, \Psi)
    ),
\end{equation}
Each of these equalities is a recurrence equation whose right-hand side is an operator acting on the functional $J_t$.
In particular, the minimum supply $J_t(P_*, \Psi)$, required to drive the system from the nominal invariant state distribution $P_*$ to a \emph{different} state distribution $\Psi$ in time $t$, is nonincreasing in $t$. Indeed, (\ref{Bellman}) implies that
$$
    J_{t+1}(P_*, \Psi)
    \<
    J_1(P_*, P_*) +
    J_t(P_*, \Psi)
    =
    J_t(P_*, \Psi)
$$
in view of (\ref{rest}).  Here, $J_t(P_*, \Psi)$ is  analogous to the {\it required supply} in the sense of \cite[Definition~5 on p.~329]{Willems_1972}. A similar monotonicity condition holds for $J_t(\Phi, P_*)$, which quantifies the cost for the noise player  to drive the state distribution of the system from $\Phi$ to $P_*$ in time $t$. Another representation of (\ref{sum}) in a form, known in the Russian optimization literature as the ``Kiev broom", ``walking tube" or ``local variation" method (see, for example, \cite{MIS_1978}), is
\begin{equation}
\label{Kiev}
    J_t(\Phi,\Psi)
    =
    \inf_{P_1, \ldots, P_{t-1}}
    \sum_{k=0}^{t-1}
    J_1(P_k, P_{k+1}),
\end{equation}
where the infimum is taken over appropriately reachable intermediate state distributions $P_1, \ldots, P_{t-1}$, with $P_0 = \Phi$ and $P_t = \Psi$. The sum on the right-hand side of (\ref{Kiev}) is the minimum conditional relative entropy supply over the time interval $[0,t)$ required to drive the system along a  specified sequence of state distributions $P_0, P_1, \ldots, P_{t-1},P_t$.  In this {\it state distribution tracking problem}, any optimal noise strategy  is Markov with respect to the state $X$ of the system. This can be verified by the argument, employed in the proof of Theorem~\ref{th:sum}, that  application of the operator (\ref{MMM})  leads to a more economical Markov noise strategy. In particular,  the minimum conditional relative entropy supply rate per time step, required to maintain the system in a fixed state distribution $\Phi$ (reachable from itself in one step),  is
\begin{equation}
\label{J1}
    t^{-1}
    \inf_{Q_{0,t}:\, P_0 = P_1= \ldots = P_t = \Phi}
    E_{0,t}
    =
    J_1(\Phi,\Phi).
\end{equation}
The fact that (\ref{J1}) holds not only in the infinite-horizon limit $t\to +\infty$ but also for any $t>0$, is closely related to the additivity of the conditional relative entropy supply for Markov noise strategies discussed  in Lemma~\ref{lem:superadd}. The quantity  $J_1(\Phi,\Phi)$ is positive if the state distribution $\Phi$ is not invariant under the nominal noise. In this case,  in order to maintain the system in $\Phi$, the noise player has to persistently deviate from the nominal noise model (\ref{Rinf}). Indeed, any optimal noise strategy in (\ref{J1}) is Markov with respect to the state of the system and is  completely specified by the conditional probability distributions $\sP_{W_k \mid X_k}$, $k=0, \ldots, t-1$. These distributions can be made identical to $\sP_{W_0 \mid X_0}$ which delivers  the minimum value $J_1(\Phi, \Phi)$ in the problem (\ref{Jt}) with $t=1$ and $\Psi=\Phi$. The resulting state sequence $X$ is a homogeneous Markov chain with an invariant measure $\Phi$ and a transition kernel which is different from $G$ in (\ref{G}).
Suppose the loss in system performance, associated with $\Phi$ being different from the nominal invariant state distribution $P_*$,  is quantified by a real-valued functional  $\Xi(P_*, \Phi)$. For example, the \emph{loss functional} $\Xi(P_*, \Phi)$ can describe the undesirable increase in a moment $\bE g(X_k)=\int_{\cX} g(x)\Phi(\rd x)$ of the system variables (specified by a function $g: \cX \to \mR_+$)
over a steady-state distribution $\Phi$ in comparison to the nominal value $\int_{\cX} g(x)P_*(\rd x)$ of this  moment under $P_*$.
Then the nonnegative quantity
\begin{equation}
\label{Z}
    Z(\gamma)
    :=
    \inf_{\Phi:\ \Xi(P_*, \Phi)\> \gamma}
    J_1(\Phi, \Phi)
\end{equation}
is the minimum cost for the noise player (in terms of the conditional relative entropy supply rate) to achieve a given level $\gamma$ of the system performance loss. Therefore, $Z(\gamma)$  can be interpreted as a robustness index for the system:  the larger $Z(\gamma)$ is, the more robust the system is with respect to the uncertain noise. A practically computable version of the \emph{robustness index} $Z(\gamma)$   in (\ref{Z}), associated with the second moments of state variables, will be considered in Section \ref{sec:robust_meas} for a class of linear systems.

\section{Reachability of Gaussian state distributions in linear systems}
\label{sec:state_dist_contr}

We will now specialize the results of the previous sections to linear systems with the state space $\cX:= \mR^n$, input space  $\cW:=\mR^m$, and dynamics
(\ref{system}) described by
\begin{equation}
\label{fAB}
    X_{k+1}
    =
    AX_k + BW_k.
\end{equation}
Here, $A\in \mR^{n\x n}$, $B \in \mR^{n\x m}$ are given matrices, and $A$ is assumed to be asymptotically stable (that is, its spectral radius
satisfies $\rho(A) <1$). Unless specified otherwise, vectors are assumed to be organized as columns.
Also, suppose the nominal
marginal distribution $R$ of the noise is the $m$-dimensional standard
normal distribution with zero mean and identity covariance matrix:
\begin{equation}
\label{GaussR}
    R
    :=
    \cN(0, I_m).
\end{equation}
Then the corresponding nominal invariant state distribution $P_*$ of the linear system (\ref{fAB})
 is
also Gaussian,
\begin{equation}
\label{GaussP*}
    P_*
    =
    \cN(
        0,
        \Gamma
    ),
\end{equation}
where the covariance matrix $\Gamma$ coincides with the infinite-horizon
reachability Gramian of the pair $(A, B)$ and satisfies an
algebraic Lyapunov equation:
\begin{equation}
\label{Gamma}
    \Gamma
    =
    \sum_{k=0}^{+\infty}
    A^k B B^{\rT} (A^k)^{\rT}
    =
    A\Gamma A^{\rT} + BB^{\rT}.
\end{equation}
The following theorem extends the condition for linear system reachability \cite{KS_1972,LS_1995}    from signals to probability distributions and is valid regardless of the asymptotic stability of the matrix $A$.   Let
\begin{equation}
\label{HH}
    \Gamma_t
    :=
    \sum_{k = 0}^{t-1}
    A^k BB^{\rT} (A^k)^{\rT}
    =
    H_t  H_t^{\rT}
\end{equation}
denote the reachability Gramian of the system (\ref{fAB}) for a finite time horizon $t>0$,
where $H_t\in \mR^{n\x mt}$ is an auxiliary matrix defined by
\begin{equation}
\label{H}
    H_t
    :=
    \begin{bmatrix}
        A^{t-1} B  &  A^{t-2} B & \ldots  & AB & B
    \end{bmatrix}.
\end{equation}

\begin{theorem}
\label{th:contr} Suppose the linear system
(\ref{fAB}) is endowed with the nominal marginal distribution (\ref{GaussR}) of the noise, and let
\begin{equation}
\label{GaussPhiPsi}
    \Phi
    :=
    \cN(\alpha, \Sigma),
    \qquad
    \Psi
    :=
    \cN(\beta, \Theta),
\end{equation}
be
any Gaussian distributions
with
covariance matrices $\Sigma\succ 0$ and $\Theta\succ 0$. Then  the state distribution of the system can be driven from $P_0 = \Phi$ to $P_t = \Psi$ by an admissible noise within any given time horizon $t\> n$ if and only if
$(A,B)$ is reachable. Moreover,  this can be carried out by using a Gaussian noise
sequence $W_{0:t-1}$ with the conditional distribution
\begin{equation}
\label{WX0}
    \sP_{
        W_{0:t-1}
        \mid
        X_0
    }
    =
    \cN
    \Big(
             H_t^{\rT}
     \Gamma_t^{-1}
     \big(
        \beta
        - A^t \alpha
        +
        (
            \sqrt{\Theta - \eps \Gamma_t}\,
            \Sigma^{-1/2}
            -A^t
        )
        (X_0-\alpha)
     \big),\
        \eps I_{mt}
    \Big).
\end{equation}
Here, $\Gamma_t$ is the $t$ step reachability Gramian with the associated matrix $H_t$ from (\ref{HH}), (\ref{H}), and $\eps$ is a real parameter satisfying
\begin{equation}
\label{eps}
    0
    <
    \eps
    \<
    1\big/\rho(\Gamma_t \Theta^{-1}).
\end{equation}
\end{theorem}

\begin{proof}
For the linear dynamics (\ref{fAB}) being considered, the $t$ step state transition map takes the form
\begin{equation}
 \label{XY}
    X_t
     =
    A^tX_0
    +
    \sum_{s=0}^{t-1}
    A^{t-s-1}B W_s
     =
     A^t X_0 + H_t W_{0:t-1}
     =
    F_t Y_{0:t-1}.
\end{equation}
Here, the state-noise sequence $Y_{0:t-1}$ and the noise sequence $W_{0:t-1}$, defined by (\ref{Yst}) and (\ref{Wst}), are organized as the vectors
\begin{equation}
\label{GaussYW}
    Y_{0:t-1}
    =
    \begin{bmatrix}
        X_0\\
        W_{0:t-1}
    \end{bmatrix},
    \qquad
    W_{0:t-1}
    =
    \begin{bmatrix}
        W_0\\
        \vdots\\
        W_{t-1}
    \end{bmatrix}
\end{equation}
of dimensions $n+mt$ and $mt$, respectively.
Accordingly, the matrix $F_t$, which  describes the linear state transition map $Y_{0:t-1}\mapsto X_t$ in (\ref{XY}), is associated with $H_t$ from (\ref{H}) by
\begin{equation}
\label{FH}
    F_t
    :=
    \begin{bmatrix}
        A^t &  H_t
    \end{bmatrix}
    =
    \begin{bmatrix}
        A^t &  A^{t-1} B & \ldots  & AB & B
    \end{bmatrix}.
\end{equation}
The linearity of the system (\ref{fAB}) allows the first two moments
of $X_t$ to be related to those of the state-noise
sequence $Y_{0:t-1}$ as
\begin{align}
\label{mean}
    \bE X_t
     & =
     F_t \bE Y_{0:t-1}
      =
     A^t \bE X_0 + H_t \bE W_{0:t-1},\\
\nonumber
     \cov(X_t)
     & =
     F_t \cov(Y_{0:t-1}) F_t^{\rT}\\
\label{cov}
     & =
    \cov(A^t X_0 + H_t \bE(W_{0:t-1}\mid X_0))
     +
    H_t \bE \cov(W_{0:t-1}\mid X_0) H_t^{\rT}.
\end{align}
Here, $\cov(\xi, \eta):= \bE(\xi\eta^{\rT}) - \bE\xi\bE\eta^{\rT}$ denotes the covariance matrix of square integrable random vectors $\xi$ and $\eta$, with $\cov(\xi):= \cov(\xi, \xi)$, and $\cov(\xi \mid \zeta):= \bE(\xi\xi^{\rT}\mid \zeta) - \bE(\xi\mid \zeta) \bE(\xi\mid \zeta)^{\rT}$ is the conditional covariance matrix of $\xi$ given another  random vector $\zeta$. Also, use is made of the ``total covariance" identity
$
    \cov(\xi)
    =
    \cov (\bE(\xi\mid \zeta))
    +
    \bE \cov(\xi\mid \zeta)
$; cf. \cite[Remark~4 on p.~214; Problem~2 on p.~83]
{Shiryaev}. Now, suppose the time horizon $t$ is fixed and satisfies $t\> n$, being otherwise arbitrary. We will construct an admissible  noise sequence $W_{0:t-1}$, which is jointly Gaussian with $X_0$ and drives
the system (\ref{fAB}) between the Gaussian state distributions $P_0=\Phi$ and $P_t=\Psi$ in (\ref{GaussPhiPsi}) with arbitrary mean values $\alpha$, $\beta$ and nonsingular
covariance matrices $\Sigma$, $\Theta$. Since $t\> n$,  the
reachability of $(A,B)$   is equivalent
to the positive definiteness of $\Gamma_t$,  the reachability Gramian in (\ref{HH}), which is equivalent to the matrix $H_t$ in (\ref{H}) being of full row rank.
By substituting the initial and terminal state mean values $\bE X_0:= \alpha$ and $\bE X_t := \beta$ into    (\ref{mean}), it follows that the noise sequence $W_{0:t-1}$ must satisfy
\begin{equation}
\label{EW}
    \beta
       =
     A^t \alpha + H_t \bE W_{0:t-1}.
\end{equation}
This equality can be fulfilled, for example, by using the following
particular mean values for the noise sequence
\begin{equation}
\label{EWmin}
    \bE W_{0:t-1}
     :=
     H_t^{\rT}
     \Gamma_t^{-1}
     (\beta - A^t \alpha ).
\end{equation}
The relation (\ref{EW}), which does not suppose the distribution  of $Y_{0:t-1}$ to be Gaussian,  shows that at the level of first moments,
the reachability of $(A,B)$ is not only sufficient but is also
necessary for the reachability of state distributions.
 Indeed, if $(A,B)$ is not reachable, then the image $\im H_t := \{H_t w:\ w \in \mR^{mt}\}$ of $\mR^{mt}$ under the linear map specified by the matrix $H_t$ is a proper
subspace of the system state space $\mR^n$. In this case, (\ref{EW}) can
not be satisfied if, for example, $\alpha = 0$ and $\beta
\not\in \im H_t$, thus proving the necessity.  We will now consider the second moments. By substituting the initial and terminal state covariance matrices $\cov(X_0) := \Sigma$ and $\cov(X_t):=\Theta$ from (\ref{GaussPhiPsi}) into (\ref{cov}), it follows that the Gaussian noise sequence $W_{0:t-1}$ being constructed must also satisfy
\begin{equation}
\label{Pit}
    \Theta
    =
    (A^t + H_tK)
    \Sigma
    (A^t + H_tK)^{\rT}
    +
    H_t L H_t^{\rT},
\end{equation}
where the matrices
\begin{equation}
\label{KL}
    K := \cov(W_{0:t-1}, X_0) \Sigma^{-1},
    \qquad
    L:=\cov(W_{0:t-1}) - \cov(W_{0:t-1}, X_0) \Sigma^{-1} \cov(X_0, W_{0:t-1}),
\end{equation}
together with $\Sigma$, parametrize the covariance matrix of the state-noise sequence $Y_{0:t-1}$ computed in  accordance with (\ref{GaussYW}) as
\begin{equation}
\label{covY}
    \cov(Y_{0:t-1})
    =
    \begin{bmatrix}
        \Sigma & \Sigma K^{\rT}\\
        K \Sigma & K \Sigma K^{\rT} + L
    \end{bmatrix}.
\end{equation}
Since $X_0$  and $W_{0:t-1}$ are jointly Gaussian by construction, the matrix $L$ in (\ref{KL}) coincides with the conditional covariance matrix $\cov(W_{0:t-1}\mid X_0)$ which  does not depend on the conditioning random vector $X_0$ in the Gaussian case. For a Gaussian state-noise sequence $Y_{0:t-1}$,  the admissibility of the noise, that is,  the conditional
absolute continuity of $\sP_{ W_{0:t-1} \mid X_0}$ in the sense
of (\ref{cond_abs_cont}), is equivalent to $L\succ 0$. The covariance condition (\ref{Pit}) is satisfied,
for example, if the matrices (\ref{KL}) are chosen as
\begin{equation}
\label{KLchoice}
    K =
    H_t^{\rT} \Gamma_t^{-1}
    \big(
        \sqrt{\Theta - \eps \Gamma_t}\,
        \Sigma^{-1/2}
        -A^t
    \big),
    \qquad
    L
    =
    \eps I_{mt}.
\end{equation}
Here, $\eps$ is a positive parameter small enough  to ensure the positive
semi-definiteness of $\Theta - \eps \Gamma_t$ for the real matrix
square root to be well-defined, which is equivalent to (\ref{eps}).  Thus,
a Gaussian noise sequence $W_{0:t-1}$ with the conditional distribution
\begin{equation}
\label{WX}
    \sP_{ W_{0:t-1} \mid X_0}
    =
    \cN(\bE(W_{0:t-1} \mid X_0), L),
    \qquad
    \bE(W_{0:t-1} \mid X_0)
    =
    \bE W_{0:t-1} + K(X_0 - \alpha )
\end{equation}
whose parameters are
given by (\ref{EWmin}) and (\ref{KLchoice}),
indeed drives the state distribution of the system from $P_0 = \Phi$ to $P_t = \Psi$ as specified in (\ref{GaussPhiPsi}). Here,  use is made of well-known results on conditional distributions for jointly Gaussian random vectors \cite{LS_2001,Shiryaev}.
Now, (\ref{WX0}) is obtained by substitution of (\ref{EWmin}) and (\ref{KLchoice}) into (\ref{WX}).
\hfill\end{proof}

\paragraph{\bf Remark} It follows from  Theorem~\ref{th:contr}, that the noise player can drive the linear system (\ref{fAB}) with a reachable pair $(A,B)$
between arbitrary nonsingular Gaussian state distributions by using Gaussian noise
sequences, provided the time horizon $t$ is not smaller than the state dimension $n$. The latter condition can be relaxed to $t\> \tau$, where
\begin{equation}
\label{tau}
    \tau := \min \big\{t> 0:\ \Gamma_t \succ 0\big\}
\end{equation}
is the first time when the matrix $H_t$ in (\ref{H}) acquires full row rank.
For example, if $n\< m$ and ${\rm rank}B = n$, then $\tau = 1$. $\blacktriangle$

Although the specific choice of a noise sequence which was made  in
the proof of Theorem~\ref{th:contr} is not unique, it turns out that the class of Gaussian noise strategies is large enough to contain
an optimal strategy for the problem (\ref{Jt}) with Gaussian boundary conditions $\Phi$ and $\Psi$, so that more general (non-Gaussian) noise strategies are not superior in this case.

\section{Computing the minimum conditional relative entropy supply for linear systems}
\label{sec:min_req_supp}

The significance of Gaussian
noise sequences for minimizing
the conditional relative entropy supply in  the case of linear dynamics (\ref{fAB}) is clarified by the following lemma. This lemma, which is provided for the sake of completeness,  is an adaptation to the  present case of the well-known results, which are closely related to the maximum entropy principle \cite{CT_2006,ME_1981}; see also, \cite[Lemma~4 on
pp.~313--314]{Petersen_2006}.

\begin{lemma}
\label{lem:Gauss} Suppose $\xi$ is a square integrable $\mR^r$-valued
random vector with an absolutely continuous probability distribution.
Then its relative entropy (\ref{bD0}) with respect to the $r$-dimensional Gaussian distribution $\cN(a, C)$
 with mean $a\in \mR^r$ and covariance matrix $C\succ 0$
satisfies
\begin{equation}
\label{Gauss}
    \bD_0(\sP_{\xi}\| \cN(a, C))
    \>
    \frac{1}{2}
    \Big(
         \|\bE \xi - a\|_{C^{-1}}^2
          +
          \underbrace{
         \Tr \chi
         -
         \ln\det\chi
         -
         r}_{\rm ``covariance"\ part}
    \Big),
    \qquad
    \chi:=
    C^{-1}
    \cov(\xi),
\end{equation}
where $\|v\|_M := \sqrt{v^{\rT}
M v}$ denotes  the Euclidean norm generated by a real positive definite
symmetric matrix $M$.  Furthermore,
the inequality (\ref{Gauss}) becomes an equality if and only if the distribution
$\sP_{\xi}$ is Gaussian.
\end{lemma}

The nonnegativeness of the covariance part of the right-hand side of
(\ref{Gauss}) follows directly from the positive definiteness of $C$ and $\cov(\xi)$,
whereby the eigenvalues $\lambda_1, \ldots, \lambda_r$ of the matrix $\chi$ are
all real and positive  \cite[Theorem~7.6.3 on p.~465]{HJ_2007}:
\begin{equation}
\label{covpartpos}
    \Tr\chi
    -
    \ln\det\chi
    -r
    =
    \sum_{k=1}^{r} (\lambda_k - \ln \lambda_k -1)
    \>
    0.
\end{equation}
%
%
This quantity vanishes if and only if  $C = \cov(\xi)$, since
$\min_{\lambda > 0} (\lambda-\ln \lambda) = 1$ is achieved only at
$\lambda = 1$. The following theorem provides a solution to the optimization problem (\ref{Jt}) with Gaussian boundary conditions.

\begin{theorem}
\label{th:min_req_supply} Suppose  the linear system (\ref{fAB}) has a reachable pair $(A,B)$, and the matrix $A$ is asymptotically stable. Then for any time
horizon $t\> n$ and any initial and terminal Gaussian state
distributions $\Phi$ and $\Psi$ in (\ref{GaussPhiPsi}) with
nonsingular covariance matrices, the minimum required conditional
relative entropy supply (\ref{Jt}) is computed as
\begin{equation}
\label{JJJ}
    J_t(\Phi, \Psi)
    =
    \big(
        \|
            \beta - A^t \alpha
        \|_{\Gamma_t^{-1}}^2
        +
        \Tr
        (
            U + V-\sqrt{I_n + 4UV}
        )
        -\ln\det \mho
    \big)\big/ 2.
\end{equation}
Here,
\begin{equation}
\label{UV}
    U
    :=
    \Gamma_t^{-1/2} A^t \Sigma (A^t)^{\rT} \Gamma_t^{-1/2},
    \qquad
    V
    :=
    \Gamma_t^{-1/2} \Theta \Gamma_t^{-1/2}
\end{equation}
are real positive semi-definite symmetric matrices (with $V\succ 0$)  defined using (\ref{HH}), and  $\mho$ is a real
positive definite symmetric matrix of order $n$ satisfying the
algebraic Riccati equation
\begin{equation}
\label{RicSigma}
    \mho + \mho U \mho = V.
\end{equation}
\end{theorem}

\begin{proof}
Suppose the system under consideration is initialized at the state distribution $P_0 = \Phi$. Then, in view of  (\ref{chain1}),  the conditional
relative entropy supply (\ref{Est}) over the time interval
$[0,t)$ takes the form
\begin{equation}
\label{Gauss_supply}
    E_{0,t}
    =
    \bD_0(Q_{0,t}\| P_*\x R^t) - \bD_0(\Phi\| P_*),
\end{equation}
where, as before, $Q_{0,t}$ is the probability distribution of the state-noise sequence
$Y_{0:t-1}$. In order to ensure the terminal condition $P_t = \Psi$,  the moments
  $\bE Y_{0:t-1}$ and $\cov(Y_{0:t-1})$  must satisfy (\ref{mean}) and (\ref{cov}).
In view of  (\ref{GaussR}) and (\ref{GaussP*}), the probability  measure $P_* \x R^t = \cN\Big(0,{\scriptsize\begin{bmatrix}\Gamma & 0\\ 0 & I_{mt}\end{bmatrix}}\Big)$ is a Gaussian distribution in $\mR^{n+mt}$ whose covariance matrix is nonsingular by the reachability of $(A,B)$.  Hence,
Lemma~\ref{lem:Gauss} implies that the minimum of $E_{0,t}$ in  (\ref{Gauss_supply}) with respect to $Q_{0,t}$ with fixed
$\bE Y_{0:t-1}$ and $\cov(Y_{0:t-1})$ is achieved at the Gaussian
distribution $\cN(\bE Y_{0:t-1}, \cov(Y_{0:t-1}))$. Also, by
Theorem~\ref{th:contr}, for Gaussian initial and
terminal state distributions (\ref{GaussPhiPsi}),  there exist
Gaussian noise sequences which drive the system from $P_0=\Phi$ to $P_t=\Psi$. Therefore,
consideration can be restricted to Gaussian state-noise sequences, so
that Lemma~\ref{lem:Gauss} reduces the computation of
$J_t(\Phi, \Psi)$ to the constrained minimization of the
conditional relative entropy
\begin{align}
\nonumber
    E_{0,t}
    & =
    \bD(
        \sP_{
            W_{0:t-1} \mid X_0
        }
        \|
        \cN(0, I_{mt})
    )\\
\nonumber
    & =
    \bE
    \big(
         |\bE(W_{0:t-1}\mid X_0)|^2
          +
         \Tr L
         -
         \ln\det L
         -
         mt
    \big)\big/ 2\\
\label{Gauss_supply1}
     & =
    \big(
        |
            \bE W_{0:t-1}
        |^2
        +
    \Tr
    (
        K\Sigma  K^{\rT}+L
    )
     -
    \ln\det
    L
    -mt
    \big)\big/ 2.
\end{align}
Here, use is made of the property that
if the state-noise sequence $Y_{0:t-1}$ is Gaussian with covariance
matrix (\ref{covY}), then the
conditional distribution $\sP_{W_{0:t-1}\mid X_0}$ is
given by (\ref{WX}),
and hence,
$$
    \bE (|\bE(W_{0:t-1} \mid X_0)|^2)
    =
    |\bE W_{0:t-1}|^2
    +
    \Tr (K\Sigma  K^{\rT}).
$$
The right-hand side of (\ref{Gauss_supply1}) is to be
minimized over the mean $\bE W_{0:t-1}$ subject to (\ref{EW}) and over the matrices  $K$ and
$L$ from (\ref{KL}) and (\ref{covY}) subject to the covariance
condition (\ref{Pit}). The constrained minimization of (\ref{Gauss_supply1}) over  $\bE
W_{0:t-1}$ subject to (\ref{EW}) can be ``decoupled" from the
minimization with respect to $K$ and $L$. By applying the
linearly constrained least squares method and recalling (\ref{FH})
and (\ref{HH}), it follows that
\begin{equation}
\label{minmean}
    \min_{\bE W_{0:t-1}\ {\rm satisfying}\ (\ref{EW})}
    |
        \bE W_{0:t-1}
    |^2
     =
    \|
        \beta  - A^t \alpha
    \|_{\Gamma_t^{-1}}^2.
\end{equation}
Here, the minimum is achieved at $\bE W_{0:t-1}$, described by
(\ref{EWmin}), which can be represented in a step-wise form as
$
    \bE W_k
    =
    B^{\rT}
    (A^{t-1-k})^{\rT}
    \Gamma_t^{-1}
    (\beta  - A^t \alpha )
$ for $k=0, \ldots, t-1$.
This can be obtained by solving
a linear-quadratic optimal control problem \cite{KS_1972,LS_1995} of minimizing the function $\sum_{k=0}^{t-1} |u_k|^2$ for the dynamical system
$\bE X_{k+1} = A \bE X_k + B u_k$ with respect to $u_k := \bE W_k$ subject to the boundary conditions $\bE X_0 = \alpha $ and $\bE X_t =
\beta $. The latter system results from averaging the linear dynamics (\ref{fAB}). We will now minimize the
remaining part
\begin{equation}
\label{KLminimization}
       \Tr
        (
           K\Sigma  K^{\rT}+L
        )
       -
       \ln\det L
\end{equation}
of (\ref{Gauss_supply1}) with respect to the matrices $K$ and $L$
subject to the covariance constraint (\ref{Pit}). Since $\Sigma\succ 0$, and $\ln\det L$ is strictly concave in $L\succ 0$ \cite[Theorem~7.6.7 on
p.~466]{HJ_2007}, the function
in (\ref{KLminimization}) is strictly
convex in $K$ and $L$. The structure of the constraint
(\ref{Pit}) allows corresponding Lagrange multipliers to be
assembled in a real symmetric $(n\x n)$-matrix $N$, so that the
Lagrange function for
minimizing (\ref{KLminimization}) subject to the
constraint (\ref{Pit}) is
\begin{align}
\nonumber
    \Lambda(K,L)
     := &
       \Tr
       (
            K\Sigma  K^{\rT}+L
       )
       -
       \ln\det L\\
\label{Lagrange}
        & -
       \Tr
       \big(
        N
        (           (A^t + H_tK)
           \Sigma
           (A^t + H_tK)^{\rT}
           +
           H_t L H_t^{\rT})
       \big).
\end{align}
Here, the last trace is the Frobenius inner product \cite{HJ_2007} of the matrix
$N$ and a real symmetric  matrix on the right-hand side of (\ref{Pit}).
The equations for the Fr\'{e}chet derivatives of
$\Lambda$ (with respect to  the matrices $K$ and $L$) to vanish are
\begin{eqnarray}
\label{nablaK}
\d_K \Lambda(K,L)
    & = &
    2
    (
        (
            I_{mt}
            -
            H_t^{\rT} N H_t
        )
        K
        -
        H_t^{\rT} N A^t
    )
    \Sigma
        =0, \\
\label{nablaL}
    \d_L \Lambda(K,L)
    & = &
    I_{mt}
    -
    H_t^{\rT} N H_t - L^{-1}
    =0,
\end{eqnarray}
where use is made of the Fr\'{e}chet  derivative $\d_L \ln \det L = L^{-1}$.
By solving (\ref{nablaL}) for $L$ and substituting the result into
(\ref{nablaK}), it follows that the stationary point of the Lagrange
function (\ref{Lagrange}) is described by
\begin{equation}
\label{KLmin}
    K = L H_t^{\rT} N A^t,
    \qquad
    L =
    (
        I_{mt}
        -
        H_t^{\rT} N H_t
    )^{-1}.
\end{equation}
Since the reachability Gramian in (\ref{HH}) satisfies $\Gamma_t\succ 0$
for $t\> n$,  the matrix inversion lemma
\cite[pp.~18--19]{HJ_2007} yields
\begin{eqnarray}
\nonumber
    S
    & := &
    H_t L H_t^{\rT}
     =
    H_t
    (
        I_{mt}
        +
        H_t^{\rT}
        (I_n-NH_tH_t^{\rT} )^{-1}
        NH_t
    )
    H_t^{\rT}\\
\label{S}
    & = &
    \Gamma_t +
    \Gamma_t
    (I_n-N\Gamma_t)^{-1}N\Gamma_t
       =
    (\Gamma_t^{-1} - N)^{-1}, \\
\label{HK}
    H_t K
    &= &
    S N A^t
    =
    (S\Gamma_t^{-1} - I_n) A^t.
\end{eqnarray}
 Hence, the covariance relation (\ref{Pit}) takes the form of an
algebraic  Riccati equation in the matrix $S$:
\begin{eqnarray}
\nonumber
    \Theta
    & = &
    (I_n + SN) A^t \Sigma  (A^t)^{\rT} (I_n + NS) + S\\
\label{RicS}
    & = &
    S\Gamma_t^{-1} A^t \Sigma  (A^t)^{\rT} \Gamma_t^{-1} S + S.
\end{eqnarray}
Since $H_t$ is of full row rank and $L\succ 0$, then (\ref{S}) implies that $S\succ 0$. In view of (\ref{UV}), left and right multiplication of both sides of (\ref{RicS}) by $\Gamma_t^{-1/2}$ leads to an equivalent
Riccati equation (\ref{RicSigma}) in the real positive definite
symmetric matrix
\begin{equation}
\label{Sigma}
    \mho
    :=
    \Gamma_t^{-1/2}
    S
    \Gamma_t^{-1/2}.
\end{equation}
Since $U\succcurlyeq $ and $V\succ 0$,
the Riccati equation (\ref{RicSigma}) has a unique solution
$\mho\succ 0$; see, for example,
\cite{LR_1995}. We will now express the minimum value of the
function (\ref{KLminimization}) in terms of $\mho$. Recall that for
conforming matrices $C$ and $D$, the matrices  $CD$ and $DC$ share nonzero eigenvalues
\cite[Theorem~1.3.20 on p.~53]{HJ_2007}. Hence, by changing
the order in which $H_t^{\rT}$ and $N H_t$ are multiplied in the
representation of the matrix $L$ in (\ref{KLmin}) and using
(\ref{HH}) and (\ref{S}), it follows that the spectrum of $L$
differs from that of
$$
    (
        I_{mt}
        -
        N H_t H_t^{\rT}
    )^{-1}
    =
    (I_n - N\Gamma_t)^{-1}
    =
    \Gamma_t^{-1} S
$$ only by ones. Since spectra are invariant under similarity transformations \cite{HJ_2007}, the eigenvalues of
$
    \Gamma_t^{-1} S
    =
    \Gamma_t^{-1/2} \mho \sqrt{\Gamma_t}
$ are identical to those of $\mho$ in (\ref{Sigma}).
Therefore,
\begin{equation}
\label{dettraceSL}
    \det L
     =
    \det \mho,\\
    \qquad
    \Tr L
      =
    \Tr
    \mho
     + mt-n.
\end{equation}
Furthermore, (\ref{nablaK}) and (\ref{HK}) imply that
$
    K
     =
    H_t^{\rT} N
    (A^t + H_t K)
     =
    H_t^{\rT} N
    S\Gamma_t^{-1}
    A^t
$,
and hence,
\begin{eqnarray}
\nonumber
    \Tr
    (
        K
        \Sigma
        K^{\rT}
    )
    & = &
    \Tr
    (
        H_t^{\rT}
        N
         S
         \Gamma_t^{-1}
         A^t
        \Sigma
        (A^t)^{\rT}
        \Gamma_t^{-1}
        S
        N
        H_t
    )\\
\nonumber
    & = &
    \Tr
    (
        \Gamma_t
        N
         S
         \Gamma_t^{-1/2}
         U
        \Gamma_t^{-1/2}
        S
        N
    )\\
\label{KK}
    & = &
    \Tr
    (
        \sqrt{\Gamma_t}
        N
        \sqrt{\Gamma_t}
         \mho
         U
         \mho
        \sqrt{\Gamma_t}
        N
        \sqrt{\Gamma_t}
    )
     =
    \Tr
    (
        \Delta
         U
         \Delta
    ).
\end{eqnarray}
Here,
\begin{equation}
\label{Delta}
    \Delta
     :=
    \sqrt{\Gamma_t}
    N
    \sqrt{\Gamma_t} \mho
     =
    \sqrt{\Gamma_t}
    (\Gamma_t^{-1} - S^{-1}) \sqrt{\Gamma_t} \mho
    =
    \mho - I_n
\end{equation}
is a real symmetric matrix associated with (\ref{Sigma}), and  use has
been made of (\ref{S}) which implies that $N = \Gamma_t^{-1} - S^{-1}$.
Now, by combining (\ref{Delta}) with the Riccati equation
(\ref{RicSigma}), it follows that
$$
    \Delta U \Delta
     =
    U - U\mho - \mho U + \mho U \mho
     =
    U - U\mho - \mho U + V-\mho,
$$
and hence, (\ref{KK}) becomes
\begin{equation}
\label{KKDUD}
    \Tr(K\Sigma  K^{\rT})
    =
    \Tr(U + V - 2 U\mho  - \mho).
\end{equation}
Furthermore, Lemma~\ref{lem:Riccati}, which will be established  in Section~\ref{sec:Ric} independently of the current proof, implies  that
\begin{equation}
\label{2USigma}
    2U\mho
    =
    4UV
    \big(
        I_n+
        \sqrt{
            I_n
            +
            4UV
        }
    \big)^{-1}
    =
    \sqrt{I_n + 4UV}-I_n.
\end{equation}
It now follows from (\ref{KKDUD}), (\ref{2USigma}) and
(\ref{dettraceSL}) that the minimum value of the function
(\ref{KLminimization}) is computed as
\begin{eqnarray}
\nonumber
    \min_{K,L\ {\rm satisfying}\ (\ref{Pit})}\!\!\!\!\!\!\!\!\! & &
    \left(
       \Tr
        (
           K\Sigma  K^{\rT}+L
        )
       -
       \ln\det L
    \right)\\
\nonumber
    & = &
    \Tr(U + V - 2 U\mho)
    - \ln\det \mho + mt- n\\
\label{KLminimization1}
    & = &
    \Tr(U + V - \sqrt{I_n + 4UV})
    - \ln\det \mho + mt.
\end{eqnarray}
Finally, (\ref{JJJ}) is obtained by substituting (\ref{minmean}) and
(\ref{KLminimization1}) into the right-hand side of
(\ref{Gauss_supply1}). \hfill\end{proof}


A closed-form solution of the Riccati equation (\ref{RicSigma}) will be provided in
Section~\ref{sec:Ric}.
The proof of
Theorem~\ref{th:min_req_supply} shows that
$
    \mho
    =
    \cov
    (
        \Gamma_t^{-1/2}X_t\mid X_0
    )
$,
is the conditional covariance matrix of
the ``balanced'' terminal state  $\Gamma_t^{-1/2} X_t$ of the system
under the optimal noise strategy on the time interval $[0,t)$ which delivers the minimum value $J_t(\Phi,\Psi)$ in the problem (\ref{Jt}).
The corresponding cross-covariance matrix of the initial and
balanced terminal states is
$
    \cov(\Gamma_t^{-1/2} X_t,\, X_0)
    =
    \mho \Gamma_t^{-1/2} A^t
    \Sigma
$.
Similarly to the inequality (\ref{covpartpos}), the ``covariance''
part of the right-hand side of (\ref{JJJ}) is always nonnegative:
$
        \Tr (U + V-\sqrt{I_n + 4UV})
         - \ln\det \mho
           =
        \Tr (\Delta + \Delta U \Delta) - \ln\det (I_n +\Delta)
        \>
        \Tr \Delta - \ln\det (I_n +\Delta)
        \> 0
$
in view of (\ref{Delta}).  It only
vanishes if the solution of the Riccati equation (\ref{RicSigma}) is
$\mho = I_n$, or equivalently, if  the matrices (\ref{UV}) satisfy
$V = I_n + U$. The latter equality holds if and only if the
initial and terminal state covariance matrices $\Sigma$ and $\Theta$ in (\ref{GaussPhiPsi}) are related by the Lyapunov equation
$$
    \Theta
    =
    A^t \Sigma  (A^t)^{\rT}
    +
    \Gamma_t.
$$
The right-hand of this equation, as a function of time $t$,  describes the evolution of the state
covariance matrix $\cov(X_t)$ which the linear system (\ref{fAB}) would have under the nominal noise, provided $\cov(X_0) = \Sigma$.
Furthermore, as $t\to +\infty$,
the minimum conditional relative entropy supply required to drive
the system to the terminal state distribution $\Psi=\cN(\beta,
\Theta)$ ceases to depend on the initial state distribution
$\Phi$ from (\ref{GaussPhiPsi}) and approaches the relative entropy
of $\Psi$ with respect to the nominal
invariant state  distribution in (\ref{GaussP*}),
\begin{equation}
\label{Jinf}
    \lim_{t\to +\infty}
    J_t(\Phi, \Psi)
     =
    \big(
        \|
            \beta
        \|_{\Gamma^{-1}}^2
        +
        \Tr (\Gamma^{-1} \Theta)
        -\ln\det (\Gamma^{-1} \Theta)-n
    \big)\big/2
     =
    \bD_0(\Psi \| P_*),
\end{equation}
where $\Gamma = \lim_{t\to +\infty} \Gamma_t$ is given by (\ref{Gamma}).
This can be obtained from (\ref{JJJ}), since $\rho(A)<1$ implies that the matrix $U$ in (\ref{UV}) vanishes
asymptotically, while both $V$ and the solution $\mho$ of the
Riccati equation converge to $\Gamma^{-1/2} \Theta \Gamma^{-1/2}$.
Since the infinite-horizon limit of $J_t(\Phi,\Psi)$ in (\ref{Jinf})
is independent of
$\Phi$, it could not be less than $\bD_0(\Psi \| P_*)$, in view of the lower bound (\ref{lower_bound}).

\paragraph{\bf Remark} In view of Lemma~\ref{lem:Gauss}, the proof of
Theorem~\ref{th:min_req_supply} shows that the right-hand side of the equality (\ref{JJJ}), which is computed in terms of the first two moments $\alpha$, $\Sigma$ and $\beta$, $\Theta$ of the initial and terminal state distributions $\Phi$ and $\Psi$, remains valid as a lower bound for $J_t(\Phi, \Psi)$ if $\Phi$ or $\Psi$ are not Gaussian. $\blacktriangle$

\section{Closed-form solution of the Riccati equation}
\label{sec:Ric}

The following lemma provides an explicit solution to the Riccati equation (\ref{RicSigma}), which will allow the result of Theorem~\ref{th:min_req_supply} to be given in a closed form.

\begin{lemma}
\label{lem:Riccati} The algebraic Riccati
equation  (\ref{RicSigma}), with $U\succcurlyeq 0$ and $V\succ 0$, has a unique positive
definite solution which is computed as
\begin{equation}
\label{SigmaUVsolve}
    \mho
    =
    2V
    \big(
        I_n+
        \sqrt{
            I_n
            +
            4UV
        }
    \big)^{-1}.
\end{equation}
\end{lemma}

\begin{proof}
Since $\mho\succ 0$, then by left multiplying both sides of (\ref{RicSigma}) by $\mho^{-1}$ and right multiplying  them by a matrix
\begin{equation}
\label{T}
    T: = \mho^{-1} V,
\end{equation}
the Riccati equation is transformed to $\mho^{-1} \mho(I_n + U\mho)  T = \mho^{-1} V T$, which is a quadratic equation  in the matrix $T$:
\begin{equation}
\label{TTUV}
    T^2-T = UV.
\end{equation}
The latter can, in principle, be solved by completing the square as $T^2-T = (T-I_n/2)^2 - I_n/4$, so that (\ref{TTUV}) yields
\begin{equation}
\label{TT}
    T
    =
    I_n/2 + \sqrt{I_n/4 + UV}
    =
    \big(I_n + \sqrt{I_n + 4UV}\big)\big/2.
\end{equation}
However, a more rigorous way to arrive at (\ref{TT}), which gives the correct meaning to the square root, is as follows. The properties  $\mho \succ 0$ and $V\succ 0$ imply that the matrix $T$ in (\ref{T}) is diagonalizable and its eigenvalues $d_1, \ldots, d_n$ are all real and
positive in view of \cite[Theorem~7.6.3 on p.~465]{HJ_2007}. Moreover,
\begin{equation}
\label{d1}
    d_k \> 1,
    \qquad
    k=1, \ldots, n.
\end{equation}
Indeed, from (\ref{RicSigma}) and the condition $U\succcurlyeq 0$, it follows that $V\succcurlyeq \mho$, and hence, $C:=\mho^{-1/2} V \mho^{-1/2}\succcurlyeq I_n$, whereby the eigenvalues of the matrix $C$ are not less than 1.  It remains to note that the matrix $T$ in (\ref{T}) is related to $C$ by a similarity transformation $T= \mho^{-1/2}\mho^{-1/2} V \mho^{-1/2}\sqrt{\mho} = \mho^{-1/2}C\sqrt{\mho}$, whereby $T$ has the same spectrum as $C$, thus proving (\ref{d1}). Due to its diagonalizability, the matrix $T$ can be represented as
\begin{equation}
\label{TEDE}
    T = EDE^{-1},
    \qquad
    D
    :=
    \diag_{1\< k \< n}(d_k),
\end{equation}
where the columns of $E$ are the corresponding
eigenvectors of $T$. Substitution of (\ref{TEDE}) into (\ref{TTUV})
yields
\begin{equation}
\label{dd}
    E\Omega E^{-1} = UV,
    \qquad
    \Omega
    :=
    D^2-D
    =
    \diag_{1\< k \< n}(\omega_k),
    \qquad
    d_k^2 - d_k = \omega_k.
\end{equation}
Hence, the columns of $E$ are also the eigenvectors of $UV$, which correspond to the eigenvalues $\omega_1, \ldots, \omega_n$. Since $UV$ is a diagonalizable matrix whose spectrum is all real and nonnegative (in view of $U\succcurlyeq 0$ and $V\succ 0$), then each
of the $n$ quadratic equations in (\ref{dd}) has a unique admissible
solution
$
    d_k = (1 + \sqrt{1+4\omega_k})/2
$ which satisfies (\ref{d1}). Substitution of these solutions into (\ref{TEDE}) yields
\begin{equation}
\label{TEE}
    T
    =
    \frac{1}{2}
    E
    \Big(
        I_n +
        \diag_{1\< k \< n}(\sqrt{1+4\omega_k})
    \Big)
    E^{-1}\\
    =
    \big(
        I_n
        +
        \sqrt{I_n + 4UV}
    \big)
    \big/2,
\end{equation}
thus proving (\ref{TT}).
The second equality from (\ref{TEE}) was used in the proof of
Theorem~\ref{th:min_req_supply} in the form of (\ref{2USigma}). Now,
(\ref{T}) allows
$\mho$ to be uniquely recovered from $T$ as $\mho =
VT^{-1}$, so that (\ref{SigmaUVsolve}) follows from (\ref{TEE}).
\hfill\end{proof}

Substitution of (\ref{SigmaUVsolve}) into (\ref{JJJ}) leads to an explicit form for  the minimum required conditional relative entropy supply, computed in Theorem \ref{th:min_req_supply}:
\begin{align}
\nonumber
    J_t(\cN(\alpha, \Sigma), \cN(\beta, \Theta))
    = &
    \big(
        \|
            \beta - A^t \alpha
        \|_{\Gamma_t^{-1}}^2
        +
        \Tr
        (
            U + V-\sqrt{I_n + 4UV}
        )\\
        &
\label{JJJ1}
        +\ln\det \big(I_n+
        \sqrt{
            I_n
            +
            4UV
        }
    \big)
    -
    \ln\det (2V)
    \big)\big/ 2,
\end{align}
where, as before,  the matrices $U$ and $V$ are given by (\ref{UV}). In the next section,   we will apply the representation (\ref{JJJ1}) to computing the robustness index $Z$ in (\ref{Z}) for the loss functional $\Xi$ associated with the second moments of the state variables.

\section{Computing the robustness index for one-step reachable linear systems}
\label{sec:robust_meas}

Suppose the state dimension of the system (\ref{fAB}) does not exceed the input dimension, that is, $n\< m$, and the matrix $B$ is of full row rank. Then the one-step reachability Gramian
\begin{equation}
\label{Gamma1}
    \Gamma_1 = BB^{\rT}
\end{equation}
from (\ref{HH}) is positive definite, so that $\tau = 1$ in (\ref{tau}). By Theorem \ref{th:min_req_supply}, the minimum conditional relative entropy supply rate $J_1(\Phi,\Phi)$, required for the noise player to maintain such a system in a state distribution  $\Phi$ with mean $\alpha \in \mR^n$ and covariance matrix $\Sigma\succ 0$,  satisfies
\begin{equation}
\label{JPhi}
    J_1(\Phi,\Phi)
    \>
    J_1(\cN(\alpha,\Sigma),\cN(\alpha,\Sigma))
    =:
    \wt{J}(\alpha,\Sigma).
\end{equation}
This inequality follows from the remark made at the end of Section~\ref{sec:min_req_supp} and becomes an equality if $\Phi$ is a Gaussian distribution.  The right-hand side of (\ref{JPhi}) is computed  by letting $t:=1$, $\beta:= \alpha$, $\Theta:=\Sigma$ in (\ref{UV}) and (\ref{JJJ1}) as
\begin{align}
\nonumber
    \wt{J}(\alpha, \Sigma)
    = &
    \big(
        \|(I_n-A) \alpha\|_{\Gamma_1^{-1}}^2
        +
        \ln\det(\Gamma_1/2)\\
\nonumber
        & +
        \Tr((A^{\rT}\Gamma_1^{-1}A + \Gamma_1^{-1})\Sigma)-\ln\det\Sigma\\
\label{Jtilde}
        &
        +\ln\det(I_n+\sqrt{I_n + 4M})
        -
        \Tr
            \sqrt{I_n + 4M}
    \big)\big/ 2,
\end{align}
where $M$ is an $(n\x n)$-matrix which depends quadratically on $\Sigma$ through the matrices $U$ and $V$ from (\ref{UV}) as
\begin{equation}
\label{MMMM}
    M
    :=
    UV
    =
    \Gamma_1^{-1/2} A \Sigma A^{\rT} \Gamma_1^{-1}
    \Sigma\Gamma_1^{-1/2}
    \qquad
    U= \Gamma_1^{-1/2} A \Sigma A^{\rT} \Gamma_1^{-1/2},
    \qquad
    V = \Gamma_1^{-1/2} \Sigma \Gamma_1^{-1/2}.
\end{equation}
Now,
consider a particular variant of the robustness index $Z$ in (\ref{Z}) associated with the following loss functional
\begin{equation}
\label{Xi}
    \Xi(P_*, \Phi)
    :
    =
    \frac{\|\alpha\|_{\Pi}^2 + \Tr(\Pi\Sigma)}{\Tr(\Pi\Gamma)},
\end{equation}
where $\alpha$ and $\Sigma$ are the mean vector and covariance matrix of the state distribution $\Phi$, which is not necessarily Gaussian. Here,
$\Pi$ is a given real positive definite symmetric matrix of order $n$, and $\Gamma$ is the infinite-horizon reachability Gramian from (\ref{Gamma}).   The numerator and denominator of the fraction in (\ref{Xi}) are the expectations $\bE (\|X_k\|_{\Pi}^2)$ of the state vector $X_k$ of the system over $\Phi$ and the nominal invariant state distribution $P_*$ from (\ref{GaussP*}), respectively,  with  $\Pi$ playing the role of a weighting matrix. It is assumed that small values of the weighted second moment of the state  variables are beneficial for system  performance under the nominal noise, so that an increase in this moment, described by  (\ref{Xi}), quantifies the deterioration of the system performance when the statistical uncertainty leads to a different steady-state distribution  $\Phi\ne P_*$. Also, $Z(\gamma)=0$ for all $\gamma\< 1$, and the robustness index $Z(\gamma)$ is positive for $\gamma >1$.  $Z(\gamma)$ will be of interest for those (sufficiently large) values of $\gamma$ which represent a ``critical'' level of system performance loss in terms of (\ref{Xi}).
Similar ideas, which are concerned  with  second moment increases  in the framework of entropy theoretic formulations of uncertainty, can be found in \cite{CR_2007,DVKS_2001,MKV_2011,Petersen_2006,SVK_1994,VDK_2006,VKS_1995,VKS_1996a}. The following theorem outlines the computation of the robustness index being considered here. Its formulation employs a function
\begin{equation}
\label{sig0}
    \sigma(z)
    :=
    \big(\ln (1 + \sqrt{1+4z})-\sqrt{1+4z}\big)'
    =
    -2/(1 + \sqrt{1+4z})
\end{equation}
of a complex variable $z$. Since $\sigma$ is analytic in a neighbourhood of $\mR_+$, then $\sigma(M)$ is well-defined  for the matrix $M$ in (\ref{MMMM}) whose eigenvalues are real and nonnegative. In fact, the function $\sigma$ was already used in this role in (\ref{SigmaUVsolve}).

\begin{theorem}
\label{th:Z}
Suppose the matrix $A$ in the linear system (\ref{fAB}) is asymptotically stable and the matrix $B$ is of full row rank, that is, ${\rm rank} B=n\< m$. Then for any $\gamma\>1$,
the robustness index (\ref{Z}), which corresponds to the loss functional (\ref{Xi}) with a weight matrix $\Pi\succ 0$, can be computed as
\begin{equation}
\label{ZZZ}
    Z(\gamma)
    =
    \wt{J}(0,\Sigma_{\lambda}).
\end{equation}
Here, $\wt{J}$ is the function, defined by (\ref{Jtilde}), and the matrix $\Sigma_{\lambda}\succ 0$ is a solution to the algebraic equation
\begin{equation}
\label{Sig}
    \Sigma_\lambda
    =
    \big(
    A^{\rT} \Gamma_1^{-1/2}
    (
    I_n
    +V
    \sigma(M))
    \Gamma_1^{-1/2} A
    +
    \Gamma_1^{-1/2}
    (I_n+\sigma(M)
    U)
    \Gamma_1^{-1/2}
    -
    \lambda \Pi
    \big)^{-1},
\end{equation}
which is defined in terms of (\ref{Gamma1}), (\ref{MMMM}), (\ref{sig0}) and depends on a scalar parameter $\lambda$ to be found from the equation
\begin{equation}
\label{lam}
    \Tr(\Pi\Sigma_\lambda)/\Tr(\Pi\Gamma) = \gamma.
\end{equation}

\end{theorem}
\begin{proof}
The loss functional $\Xi(P_*, \Phi)$ in (\ref{Xi}) depends on the state distribution $\Phi$ only through its first two moments $\alpha$ and $\Sigma$  and so does the right-hand side of the inequality in (\ref{JPhi}) which is achieved  for Gaussian state distributions $\Phi$. Hence, the minimization in (\ref{Z}) can be reduced to the class of  Gaussian distributions $\Phi$ without affecting the minimum value. This allows the robustness index $Z(\gamma)$, which corresponds to (\ref{Xi}), to be computed by solving a  constrained optimization problem
\begin{equation}
\label{ZZ}
    Z(\gamma)
    =
    \min
    \big\{
        \wt{J}(\alpha, \Sigma):\
        \alpha\in \mR^n,\
        \Sigma\succ 0,\
        \wt{\Xi}(\alpha,\Sigma)\> \gamma \Tr (\Pi\Gamma)/2
    \big\},
\end{equation}
where
\begin{equation}
\label{Xitilde}
    \wt{\Xi}(\alpha, \Sigma):= (\|\alpha\|_{\Pi}^2 + \Tr(\Pi\Sigma))/2,
\end{equation}
and the $1/2$ factor is introduced  for  the sake of convenience. In view of (\ref{Jtilde}) and (\ref{Xitilde}), the Lagrange function for the constrained minimization problem (\ref{ZZ}) takes the form
\begin{align}
\nonumber
    \Ups(\alpha, \Sigma)
    := &
    \wt{J}(\alpha, \Sigma) - \lambda \wt{\Xi}(\alpha, \Sigma)\\
\nonumber
    = &
    \big(
        \|\alpha\|_{(I_n-A^{\rT})\Gamma_1^{-1}(I_n-A)-\lambda \Pi }^2
        +
        \ln\det(\Gamma_1/2)\\
\nonumber
        & +
        \Tr((A^{\rT}\Gamma_1^{-1}A + \Gamma_1^{-1}-\lambda \Pi)\Sigma)-\ln\det\Sigma\\
\label{Ups}
        &
        \ln\det(I_n+\sqrt{I_n + 4M})
        -
        \Tr
            \sqrt{I_n + 4M}
    \big)\big/ 2,
\end{align}
where $\lambda \in \mR$ is a Lagrange multiplier. The dependence of the Lagrange function $\Ups$ on $\alpha$ is quadratic and can be decoupled from the dependence on $\Sigma$. The corresponding quadratic form is positive definite if and only if
$$
    \lambda < 1/\rho(\Pi (I_n-A)^{-1}\Gamma_1(I_n-A^{\rT})^{-1}).
$$
In this case, $\min_{\alpha \in \mR^n}\Ups(\alpha, \Sigma)$ is achieved at the unique point $\alpha=0$, so that the minimization of the Lagrange function $\Ups$ in (\ref{Ups}) reduces to
\begin{equation}
\label{Upsmin}
    \min_{\alpha\in \mR^n,\ \Sigma\succ 0}
    \Ups(\alpha, \Sigma)
    =
    \min_{\Sigma\succ 0}
    \Ups(0, \Sigma).
\end{equation}
We will now find a stationary point of the function $\Ups(0,\Sigma)$.
In view of the identity $\ln\det N=\Tr \ln N$ for a matrix $N$ with positive real spectrum, the application of  \cite[Lemma 4]{VP_2010a} (see also \cite[p. 270]{SIG_1998}) yields the following first variation
\begin{equation}
\label{delta}
    \delta
    \big(
        \ln\det (I_n + \sqrt{I_n + 4M})
        -
        \Tr\sqrt{I_n + 4M}
    \big)
    =
    \Tr(\sigma(M)\delta M),
\end{equation}
where the function $\sigma$ is defined by (\ref{sig0}).
Since the first variation of the map $\Sigma \mapsto M$, described by (\ref{MMMM}), is
$$
    \delta M
    =
    \Gamma_1^{-1/2} A
    (
        (\delta \Sigma) A^{\rT} \Gamma_1^{-1}\Sigma
        +
        \Sigma A^{\rT} \Gamma_1^{-1}\delta \Sigma
    )
    \Gamma_1^{-1/2},
$$
then the Fr\'{e}chet derivative of the function in (\ref{delta}), as a composite function of the matrix $\Sigma$, can be computed as
\begin{align}
\nonumber
    \d_{\Sigma}
    \big(
        \ln\det (I_n + \sqrt{I_n + 4M})&
        -
        \Tr\sqrt{I_n + 4M}
    \big)\\
\nonumber
    =&
    A^{\rT} \Gamma_1^{-1}\Sigma\Gamma_1^{-1/2}
    \sigma(M)
    \Gamma_1^{-1/2} A
    +
    \Gamma_1^{-1/2}
    \sigma(M)
    \Gamma_1^{-1/2} A\Sigma A^{\rT} \Gamma_1^{-1}\\
\label{dSigma}
    =&
    A^{\rT} \Gamma_1^{-1/2}
    V
    \sigma(M)
    \Gamma_1^{-1/2} A
    +
    \Gamma_1^{-1/2}
    \sigma(M)
    U
    \Gamma_1^{-1/2}.
\end{align}
The right-hand side of (\ref{dSigma}) is a real symmetric matrix, which inherits its symmetry from $\Sigma$ in view of the identities $\sigma(UV)U = U\sigma(VU)$ and $V\sigma(UV) = \sigma(VU)V$ and the symmetry of the matrices $U$ and $V$  in (\ref{MMMM}). From (\ref{dSigma}), it follows that the equation $\d_{\Sigma}\Ups(0,\Sigma)=0$ for a stationary point $\Sigma$ of the Lagrange function (\ref{Ups}) in the minimization problem (\ref{Upsmin}) takes the form
$$
    A^{\rT}\Gamma_1^{-1}A + \Gamma_1^{-1}-\lambda \Pi
    -\Sigma^{-1}
    +
    A^{\rT} \Gamma_1^{-1/2}
    V
    \sigma(M)
    \Gamma_1^{-1/2} A
    +
    \Gamma_1^{-1/2}
    \sigma(M)
    U
    \Gamma_1^{-1/2}
    =0,
$$
which is equivalent to (\ref{Sig}).  The solution $\Sigma_{\lambda}$ of this equation depends on the Lagrange multiplier $\lambda$, which, by the standard procedure,  is to be found from (\ref{lam}) in accordance with the constraint in (\ref{ZZ}). \hfill\end{proof}

Note that (\ref{Sig}) and (\ref{lam}) form a complete set of equations for finding the pair $(\lambda,\Sigma_{\lambda})$ for a given $\gamma\> 1$.  In particular, the solution of these equations for $\gamma=1$ is $\lambda=0$ and $\Sigma_0 = \Gamma$, which corresponds to the nominal noise model, with $Z(1)=0$.  Properties of the solution for $\gamma>1$, including existence and uniqueness, require additional investigation and will be discussed elsewhere. A numerical scheme for solving (\ref{Sig})--(\ref{lam}) for $\gamma>1$ can be based on the ideas of homotopy methods, whereby (\ref{Sig}) is solved iteratively for gradually increasing values of the Lagrange multiplier $\lambda$ starting from $\lambda=0$. A closed-form calculation of the robustness index for a one-dimensional example is given in the next section.

\section{Illustrative example: one-dimensional linear systems}
\label{sec:example}

In order to avoid reachability issues for short time horizons $t$, which are associated with the condition $t\> n$ in Theorems~\ref{th:contr} and \ref{th:min_req_supply} (or its refined version $t\> \tau$ based on (\ref{tau})),
consider the one-dimensional case  $n=m=1$. Here, both
$A$ and $B$ in (\ref{fAB}) are scalars, with $|A| < 1$ and
$B\ne 0$, and the nominal marginal distribution $R$ of the noise in
(\ref{GaussR}) is $\cN(0,1)$. In this case, the variance of
the nominal invariant state distribution $P_*$ in (\ref{GaussP*}) is
\begin{equation}
\label{alphaPi}
    \Gamma
    =
    \frac{B^2}{1-A^2}.
\end{equation}
The equations (\ref{HH}) and (\ref{UV})
give
\begin{equation}
\label{muGammaUV}
    \Gamma_t
    =
    (1-A^{2t}) \Gamma,
    \qquad
    U   =   \frac{A^{2t} \Sigma }{\Gamma_t},
    \qquad
    V = \frac{\Theta }{\Gamma_t}.
\end{equation}
The solution (\ref{SigmaUVsolve})
of the Riccati equation (\ref{RicSigma}) takes the form
\begin{equation}
\label{Sigmaexample}
    \mho
     =
              \frac{2V}
             {1 + \sqrt{1+4UV}}
             =
         \frac{2\Theta}{\Gamma_t +
         \sqrt{\Gamma_t^2+4A^{2t}\Sigma  \Theta }}.
\end{equation}
By substituting these formulae
into (\ref{JJJ}) or (\ref{JJJ1}), it follows that the minimum required conditional relative
entropy supply for the noise player to drive the system from an initial state
distribution $\Phi:= \cN(\alpha , \Sigma )$ to a terminal state
distribution $\Psi:=\cN(\beta , \Theta )$ (both with positive
variances $\Sigma $ and $\Theta $) in a given time $t$  is
\begin{equation}
\label{JJJ2}
    J_t(\Phi, \Psi)
    =
    \frac{1}{2}
    \left(
        \frac{
            (\beta  - A^t \alpha )^2
            +
            A^{2t} \Sigma  + \Theta
            -\sqrt{\Gamma_t^2 + 4A^{2t} \Sigma  \Theta }
        }
        {\Gamma_t}
        -\ln\mho
    \right).
\end{equation}
The minimum conditional relative entropy supply rate $\wt{J}(\alpha,\Sigma)$ in (\ref{JPhi}),
required to maintain the system in the fixed Gaussian state distribution
$\cN(\alpha,\Sigma)$,  is calculated by letting $t:= 1$, $\beta :=\alpha $,
$\Theta  := \Sigma $ in (\ref{muGammaUV})--(\ref{JJJ2}) which yields
\begin{equation}
\label{JJJ3}
    \wt{J}(\alpha,\Sigma)
    =
    \frac{1}{2}
    \left(
        \frac{1-A}{1+A}\,
        \frac{\alpha^2}{\Gamma}
        +
        \frac{1+A^2}{1-A^2}\gamma
        -
        \sqrt{
            1 +
            \left(
                \frac{2A\gamma}{1-A^2}
            \right)^2
        }
        -\ln\mho
    \right),
\end{equation}
where
\begin{equation}
\label{Sigmaexample1}
    \mho
    =
    \frac{2\gamma}{1-A^2 + \sqrt{(1-A^2)^2 + 4A^2 \gamma^2}},
    \qquad
    \gamma := \frac{\Sigma }{\Gamma}.
\end{equation}
The discrepancy between $\cN(\alpha, \Sigma)$ and the nominal invariant
state distribution  $P_*=\cN(0, \Gamma)$ with variance
(\ref{alphaPi}) enters (\ref{JJJ3}) only through $\alpha^2/\Gamma$ and the variance ratio $\gamma$ in (\ref{Sigmaexample1}).
In this one-dimensional case, the weight $\Pi$ in the loss functional (\ref{Xi}) can be cancelled out and the functional  takes the form
\begin{equation}
\label{Xi1}
    \Xi(P_*, \Phi)
    =
    \frac{\alpha^2 + \Sigma}{\Gamma}= \frac{\alpha^2}{\Gamma} + \gamma.
\end{equation}
In view of (\ref{ZZZ}) in Theorem~\ref{th:Z},  the robustness index (\ref{Z}), which corresponds to (\ref{Xi1}),  reduces to $\wt{J}(0,\Sigma)$ and is computed  by letting
$\alpha:=0$ in (\ref{JJJ3}):
\begin{equation}
\label{J1cov}
    Z(\gamma)
    =
    \frac{1}{2}
    \left(
        \frac{1+A^2}{1-A^2}\gamma
        -
        \sqrt{
            1 +
            \left(
                \frac{2A\gamma}{1-A^2}
            \right)^2
        }
        -\ln\mho
    \right),
    \qquad
    \gamma \> 1;
\end{equation}
see Fig.~\ref{fig:Jcov}. Note that $Z(\gamma)$
\begin{figure}[htbp]
\centering
\includegraphics[width=8cm]{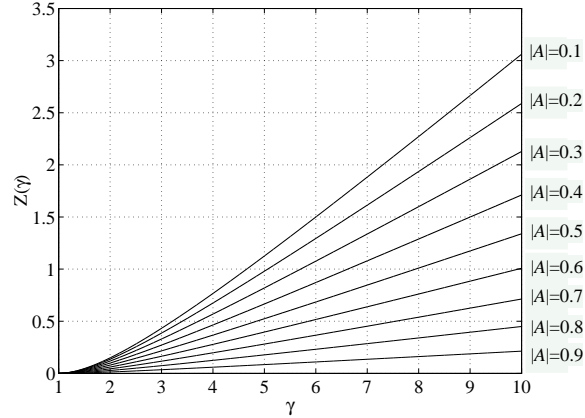}
\caption{
    The graphs of the robustness index $Z(\gamma)$ in (\ref{J1cov}), as a
    function of the variance ratio $\gamma$ from (\ref{Sigmaexample1}), for the one-dimensional linear system (\ref{fAB}) with $|A| =
    0.1, \ldots, 0.9$. Its asymptotic behavior is
    $Z(\gamma)\sim
    (1-|A|)\gamma/(2(1+|A|))$ as $\gamma \to +\infty$.
}
    \label{fig:Jcov}
\end{figure}
vanishes for $\gamma=1$ and is strictly decreasing in $|A|$ for any variance ratio $\gamma> 1$.
That is, the less stable the system is, the easier it is  for the
noise player (in the sense of the minimum required conditional
relative entropy supply rate) to maintain the system  in a state
distribution $\Phi$ with a given larger variance compared to the nominal invariant state
distribution $P_*$. This is in agreement with the intuitive expectation  that the deviation of the system from
the nominal behavior can be achieved by smaller
deviations of the noise from its nominal model
since their accumulation is more efficient if the system is less
stable.

\section*{Acknowledgments}

The work is supported by the Australian Research Council. The first author also thanks Valery A. Ugrinovskii for helpful discussions. Comments of anonymous reviewers are also gratefully acknowledged.


\end{document}